\documentclass[11pt]{article}

\usepackage[letterpaper, margin=1in]{geometry}
\usepackage{amsmath,amssymb,amsthm}
\usepackage{mathtools}
\usepackage{graphicx}
\usepackage{enumitem}
\usepackage{thm-restate}
\usepackage{hyperref}
\usepackage[nameinlink,capitalize]{cleveref}
\usepackage[linesnumbered,ruled,vlined]{algorithm2e}
\SetAlgoHangIndent{0pt}
\SetKwInOut{KwAssume}{Assumption}
\usepackage{lineno}
\usepackage[T1]{fontenc}
\usepackage[utf8]{inputenc}
\usepackage[
  backend=biber,
  doi=true,
  url=false,
  maxnames=99,
]{biblatex}
\usepackage{tabularx}
\usepackage{booktabs}
\usepackage{array}
\usepackage{makecell}
\newcolumntype{C}[1]{>{\centering\arraybackslash}m{#1}}
\newcolumntype{Y}{>{\centering\arraybackslash}X}
\declaretheorem[numberwithin=section]{theorem}
\declaretheorem[sibling=theorem]{lemma}
\declaretheorem[sibling=theorem]{claim}
\declaretheorem[sibling=theorem]{fact}

\declaretheoremstyle[
  bodyfont=\normalfont,
]{definitionstyle}

\declaretheorem[
  sibling=theorem,
  style=definitionstyle,
]{definition}

\declaretheorem[
  sibling=theorem,
  style=definitionstyle,
]{problem}



\crefname{theorem}{Theorem}{Theorems}
\crefname{lemma}{Lemma}{Lemmas}
\crefname{claim}{Claim}{Claims}
\crefname{fact}{Fact}{Facts}
\crefname{definition}{Definition}{Definitions}
\crefname{problem}{Problem}{Problems}
\crefname{appendix}{Appendix}{Appendices}
\crefname{algocf}{Algorithm}{Algorithms}

\mathtoolsset{showonlyrefs=true}

\title{Faster Exact Algorithms for Equal-Subset-Sum}

\author{
Ryosuke Yamano\thanks{
Department of Computer Science, Graduate School of Information Science and Technology, The University of Tokyo, Japan.
ryoyamano15@g.ecc.u-tokyo.ac.jp.
}
\and
Tetsuo Shibuya\thanks{
Division of Medical Data Informatics, Human Genome Center, Institute of Medical Science, The University of Tokyo, Japan.
tshibuya@hgc.jp.
Supported by KAKENHI Grant Number 23K28035.
}}

\date{}

\begin{document}

\maketitle

\begin{abstract}
We study exact algorithms for Equal-Subset-Sum in the worst-case setting:
given a set $S$ of $n$ integers, find two distinct subsets
$A,B\subseteq S$ whose sums are equal.
We establish a new state-of-the-art bound for this problem by improving
the fastest known algorithm, due to Randolph and Węgrzycki (STOC 2026),
from $O^*(1.7067^n)$ time and space to an algorithm that runs in
$O^*(1.6994^n)$ time and uses $O^*(1.5664^n)$ space.
We also improve the best known polynomial-space running time, due to
Mucha, Nederlof, Pawlewicz, and Węgrzycki (ESA 2019), from
$O^*(2.6817^n)$ to $O^*(2.5430^n)$.
Finally, we investigate time-space tradeoffs for this problem and improve
the running times achievable under a broad range of exponential-space
bounds.
\end{abstract}

\thispagestyle{empty}
\clearpage
\setcounter{page}{1}

\section{Introduction}

Subset-Sum is a fundamental NP-hard problem, which can be formulated as follows.

\begin{problem}[Subset-Sum]
Given a multiset $S$ of $n$ integers and a target integer $t$, output a
subset $A\subseteq S$ such that $t=\sum_{a\in A}a$, if such a subset
exists.
\end{problem}

Subset-Sum admits a simple meet-in-the-middle algorithm, due to Horowitz
and Sahni \cite{horowitz1974computing}, that runs in $O^*(2^{n/2})$ time \footnote{The $O^*(\cdot)$ notation suppresses $\mathrm{poly}(n)$ factors.}.
This remains the standard running-time bound for exact algorithms for
Subset-Sum, and the fastest known algorithm
\cite{chen_et_al:LIPIcs.APPROX/RANDOM.2023.39} improves over the
$O(2^{n/2})$ bound only by a polynomial factor.
Whether there exists an algorithm for Subset-Sum that runs in
$O^*(2^{(0.5-\delta)n})$ time for some constant $\delta>0$ remains a
long-standing open question.

Closely related to this open problem, Howgrave-Graham and Joux
\cite{howgrave2010new} broke this meet-in-the-middle barrier in the
\emph{average-case} setting by giving an $O^*(2^{0.337n})$-time
algorithm, which was later improved to $O^*(2^{0.283n})$ by
\cite{becker2011improved,bonnetain2020improved}.
The algorithm of Howgrave-Graham and Joux introduced the representation
technique, which was later also used for Equal-Subset-Sum (ESS), an important variant of Subset-Sum that can be
formulated as follows.

\begin{problem}[Equal-Subset-Sum (ESS) \cite{WOEGINGER1992299}]
Given a set $S$ of $n$ integers, output two distinct subsets
$A,B\subseteq S$ such that $\sum_{a\in A}a=\sum_{b\in B}b$, if such
subsets exist.
\end{problem}

ESS admits a simple meet-in-the-middle algorithm that runs in
$O^*(3^{n/2})$ time.
Surprisingly, the authors of \cite{mucha_et_al:LIPIcs.ESA.2019.73} broke
this meet-in-the-middle barrier for ESS in the \emph{worst-case} setting
by exploiting the representation technique, giving an algorithm that runs
in $O^*(1.7088^n)$ time and space.
In the \emph{average-case} setting, the authors of \cite{chen2022average}
gave an $O^*(3^{0.387n})\le O^*(1.5299^n)$-time algorithm for ESS.
The open question of whether the running time for ESS in the worst-case
setting can be improved beyond $O^*(1.7088^n)$
\cite[Section~5, Question~1]{jin_et_al:LIPIcs.ESA.2025.86}
was answered affirmatively very recently by Randolph and Węgrzycki
\cite{STOC2026.Randolph}, who gave an algorithm that runs in
$O^*(1.7067^n)$ time and space.

\subsection{Our Results}

In this work, we focus on exact algorithms for ESS in the worst-case
setting.
All our algorithms are Monte Carlo algorithms that never return false
positives, and their error probabilities can be reduced to
$2^{-\Omega(n)}$ by repetition.
Our first result establishes a new state-of-the-art algorithm for ESS by further improving the result of \cite{STOC2026.Randolph}.
While the running-time improvement is our primary contribution, the
algorithm also reduces the space usage by a much larger exponential
factor.

\begin{restatable}{theorem}{FastExpSpaceThm}
\label{thm:MainResult-ExpSpace}
There exists a Monte Carlo algorithm for ESS that runs in
$O^*(1.6994^n)$ time and uses $O^*(1.5664^n)$ space.
\end{restatable}

Our second result improves the polynomial-space running time for ESS.
We improve upon the $O^*(2.6817^n)$-time algorithm of
\cite{mucha_et_al:LIPIcs.ESA.2019.73} by analyzing two approaches suited
to different parameter regimes: one based on recent improvements for
low-space Element Distinctness \cite{chen2022truly}, and the other using
the fast polynomial-space Subset-Sum algorithm of
\cite{Bansal2018-FastPolySpaceSubsetSum} as a subroutine.

\begin{restatable}{theorem}{FastPolySpaceThm}
\label{thm:MainResult-PolySpace}
There exists a Monte Carlo algorithm for ESS that runs in
$O^*(2.5430^n)$ time and uses polynomial space.
\end{restatable}

For Subset-Sum, Schroeppel and Shamir \cite{SchroeppelAndShamir1981} gave a time-space tradeoff $T=O^*(2^n/S^2)$ for any space bound $S\le O^*(2^{n/4})$. This tradeoff was later improved for almost all choices of the tradeoff parameter by \cite{austrin2013space}. In contrast, for ESS, the analogous Schroeppel--Shamir-type tradeoff, which gives $T=O^*(3^n/S^2)$ for any space bound $S\le O^*(3^{n/4})$ \cite[Appendix~D, full version]{mucha_et_al:LIPIcs.ESA.2019.73}, has remained the only general time-space tradeoff known for the problem. As our third result, we give an improved time-space tradeoff curve for ESS, improving the running time under a broad range of exponential-space bounds.

\begin{restatable}{theorem}{TimeSpaceTradeoffThm}
\label{thm:MainResult-TimeSpaceTradeoff}
For every constant $\alpha$ satisfying $0\le\alpha\le1$, there is a
Monte Carlo algorithm for ESS that runs in
$O^*(2^{\mathcal{T}(\alpha)n})$ time and uses
$O^*(2^{\alpha n})$ space, where $\mathcal{T}(\alpha)$ is defined in
\eqref{eq:mathcalT} of \cref{sec:CaseDistinctionOverview}.
\end{restatable}

In \cref{fig:timespace-tradeoff-curves}, we compare $\mathcal{T}(\alpha)$ with the best bound previously obtainable for each space exponent $\alpha$. This previous bound is obtained by taking the best among the Schroeppel--Shamir-type tradeoff, the polynomial-space $O^*(2^{1.42312n})\le O^*(2.6817^n)$-time algorithm of \cite{mucha_et_al:LIPIcs.ESA.2019.73}, and the $O^*(2^{0.77117n})\le O^*(1.7067^n)$-time and space algorithm of \cite{STOC2026.Randolph}.

\begin{figure}[t]
    \centering
    \includegraphics[width=0.9\linewidth]{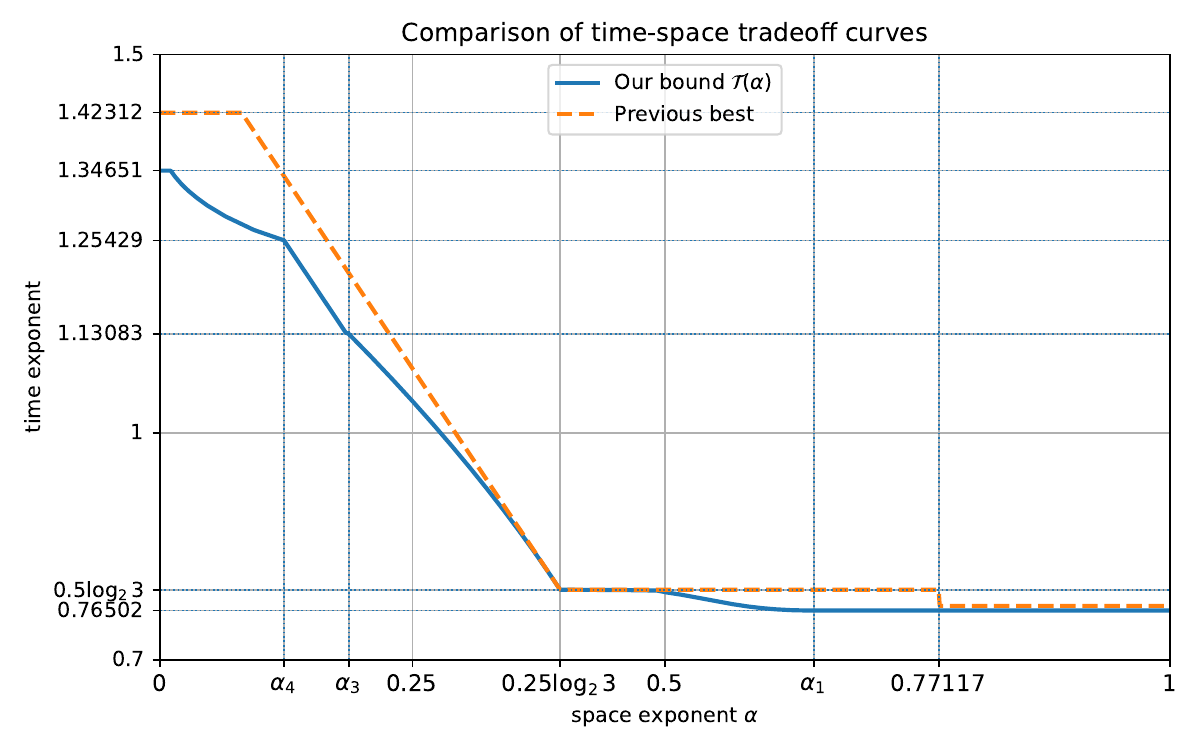}
\caption{
For each constant $\alpha\in[0,1]$, the figure plots the time exponent
$f(\alpha)$ of the running time $O^*(2^{f(\alpha)n})$ achievable under
the space bound $O^*(2^{\alpha n})$.
We compare our tradeoff curve $\mathcal{T}(\alpha)$ with the previous
best tradeoff curve.
The blue solid curve represents our bound $\mathcal{T}(\alpha)$, and the
orange dotted curve represents the previous best bound.
Here, $\alpha_1\approx0.64738$ is defined in \eqref{eq:alpha_1}, and
$\alpha_3\approx0.18687$ and $\alpha_4\approx0.12286$ are defined in
\cref{sec:CaseDistinctionOverview}.
}
    \label{fig:timespace-tradeoff-curves}
\end{figure}

\subsection{Our Main Techniques}

The algorithms of \cite{mucha_et_al:LIPIcs.ESA.2019.73} focus on the
solution size $\ell=|A|+|B|$, where $A$ and $B$ form a disjoint solution
of ESS.
They gave an algorithm that runs in $O^*(2^\ell)$ time and space for
$\ell>n/2$, which is fast when $\ell$ is small.
In our algorithm, we introduce a simple random modification to the input
instance; see \cref{alg:RandomReplace} in
\cref{sec:FasterExpSpaceAlgorithms}.
This modification reduces the solution size $\ell$ with
inverse-exponential probability.
Although we need to pay an additional exponential cost to successfully
reduce $\ell$, the exponential speedup obtained from the reduction in
$\ell$ dominates this cost in the bottleneck case
$\ell\approx 0.773n$ of \cite{mucha_et_al:LIPIcs.ESA.2019.73}.
Thus, by combining this input modification with the $O^*(2^\ell)$-time
and space algorithm of \cite{mucha_et_al:LIPIcs.ESA.2019.73}, we obtain
\cref{thm:MainResult-ExpSpace}.
Since the modified instance may have ESS solutions that do not correspond
to any solution of the original instance, we carefully filter the
candidates to avoid returning such spurious solutions.

Our modification can also be viewed as a time-space tradeoff.
It pays an additional exponential cost in time to reduce $\ell$; since
the underlying algorithm uses $O^*(2^\ell)$ space, this reduction in
$\ell$ translates directly into an exponential saving in space.
We exploit this idea in our time-space tradeoff algorithms in
\cref{sec:5-2}.
By choosing the size of the modification appropriately, our algorithm can
be used under a space bound of $O^*(2^{\alpha n})$ for a constant
$\alpha$, especially when $\alpha > 1/2$.
For more restrictive space bounds, we modify the algorithm of
\cite{mucha_et_al:LIPIcs.ESA.2019.73} by replacing the explicit
enumeration of all candidates with random access to the candidate set,
implemented via the \emph{Fast Subset-Sum Oracle} technique of
\cite{allcock_et_al:LIPIcs.ESA.2022.6}.
This allows us to apply the time-space tradeoff for Element Distinctness
\cite{lyu2023time}; see \cref{sec:5-3} for further details.
By combining these algorithms with several other algorithms
tailored to different parameter regimes, we obtain the full time-space tradeoff curve stated in \cref{thm:MainResult-TimeSpaceTradeoff}.
The detailed case distinction is given later in \cref{tab:case-distinction} of \cref{sec:CaseDistinctionOverview}.

\subsection{Related Works}
The recent attention to exact algorithms for ESS has also led to improvements
for the closely related Pigeonhole ESS problem \cite{PAPADIMITRIOU1994498}, a constrained variant of ESS in
which all input integers are positive and their sum is at most $2^n-2$.
By the pigeonhole principle, a solution is guaranteed to exist.
Thus, Pigeonhole ESS is a \emph{total search problem}, and it has also
been studied in the TFNP context \cite{BAN201948,8555101}.
Pigeonhole ESS admits a meet-in-the-middle-based algorithm that runs in $O^*(2^{n/2})$ time, which was
recently improved to $O^*(2^{n/3})$ time
\cite{jin_et_al:LIPIcs.ICALP.2024.94,jin_et_al:LIPIcs.ESA.2025.86}.

A natural optimization variant of ESS is the Subset-Sum-Ratio problem
\cite{WOEGINGER1992299}, in which the task is to find two disjoint
subsets whose sums have ratio as close to $1$ as possible.
A $1.324$-approximation algorithm was given in
\cite{WOEGINGER1992299}, and the first FPTAS for the problem was given in
\cite{BAZGAN2002160}.
This FPTAS was later simplified in \cite{nanongkai2013simple}.
Subsequent works improved the running time
\cite{melissinos2018faster,alonistiotis2024approximating,bringmann2024approximating}.

ESS also has connections to computational biology
\cite{cieliebak2003noisy,cieliebak2004measurement}.
The authors of \cite{cieliebak2003composing} considered $k$-ESS, in
which the task is to find $k$ disjoint subsets whose sums are all equal.
Other variants of ESS were studied in \cite{cieliebak2008complexity},
which also established their NP-hardness.

\section{Preliminaries} \label{sec:prelim}
We use the asymptotic notation $O^*(\cdot)$, $\Omega^*(\cdot)$, and
$\Theta^*(\cdot)$ to suppress polynomial factors in $n$, where $n$
denotes the number of integers in an ESS instance.
More precisely, $f(n)=O^*(g(n))$ if $f(n)\le g(n)\cdot n^{O(1)}$, and
$f(n)=\Omega^*(g(n))$ if $f(n)\ge g(n)\cdot n^{-O(1)}$.
We write $f(n)=\Theta^*(g(n))$ if both $f(n)=O^*(g(n))$ and
$f(n)=\Omega^*(g(n))$ hold.
Throughout the paper, all logarithms are base~2.
The binary entropy function is defined by
$H(x) = -x \log x - (1-x)\log(1-x)$ for $x \in (0,1)$,
with $H(0) = H(1) = 0$.
We write $[n]$ to denote the set $\{1, \dots, n\}$.
For a set of integers $T$, we write
$\sum(T) = \sum_{t \in T} t$ for its sum.
For integers $p$, $u$, and $v$, we write
$u \equiv_p v$ to denote $u \equiv v \pmod{p}$.
We denote by $X \Delta Y$ the symmetric difference of two sets $X$ and $Y$; that is, $X \Delta Y = (X \setminus Y) \cup (Y \setminus X)$.

By Monte Carlo algorithms, we mean randomized algorithms
that may produce only false negatives with constant probability.
The error probability can be reduced to $2^{-\Omega(n)}$
by repeating the algorithm $\mathrm{poly}(n)$ times,
without affecting the running times expressed using
$O^*(\cdot)$ notation.
We frequently use the following well-known inequalities.

\begin{fact}
For any integers $n$ and $x$ with $0\le x\le n$, we have $\frac{1}{n+1}2^{H(x/n)n} \le \binom{n}{x} \le 2^{H(x/n)n}$.
In particular, $\binom{n}{x}=\Theta^*(2^{H(x/n)n})$.
\end{fact}

We may also use the following fact to remove floors and ceilings in the
$O^*(\cdot)$ notation.

\begin{fact}[{\cite[eq.(4)]{zhang2007estimating}}]
\label{fact:H(x+1/n)-H(x)}
For any $x,y\in[0,1]$, we have $|H(x)-H(y)|\le H(|x-y|)$.
Consequently, if $x,y\in[0,1]$ satisfy $|x-y|\le O(1/n)$, then $|H(x)-H(y)| \le O((\log n)/n)$.
\end{fact}

We adopt the input assumptions from
\cite{mucha_et_al:LIPIcs.ESA.2019.73}: the input set $S$ consists of
positive integers and satisfies $\sum(S)\le 2^{O(n)}$.
We enforce these assumptions using the following preprocessing lemma,
which is based on
\cite[Appendix~A of the full version]{mucha_et_al:LIPIcs.ESA.2019.73}.
For completeness, we provide a proof in Appendix~\ref{app:A} of this paper.

\begin{restatable}{lemma}{PreprocessInputs}
\label{lm:preprocess-inputs}
Given a set $S$ of $n$ integers satisfying
$|s|\le 2^m$ for all $s\in S$, the following holds.
If $0\in S$ or $m\ge 2^n$, then ESS on $S$ can be solved in
$O(\mathrm{poly}(n,m))$ time.
Otherwise, with probability at least $1-O(2^{-n})$, we can construct, in
$O(\mathrm{poly}(n,m))$ time, a set $S'$ of
$n+\lceil\log n\rceil$ positive integers satisfying
$0<s'\le 2^{O(n)}$ for all $s'\in S'$, and satisfying the following
properties:
if $S'$ has no ESS solution, then neither does $S$; conversely, from any
ESS solution of $S'$, we can reconstruct a corresponding ESS solution of
$S$ in $O(\mathrm{poly}(n,m))$ time.
\end{restatable}

\Cref{lm:preprocess-inputs} increases the input size by only
$O(\log n)$.
Thus, a running time of the form $O^*(2^{\Theta(n)})$ changes only by a
factor of $2^{O(\log n)}=\mathrm{poly}(n)$, which is absorbed in the
$O^*(\cdot)$ notation.

For simplicity, we also assume that $n$ is divisible by $12$.
This can be achieved by setting $M=\sum(S)+1$ and adding $O(1)$ elements
of the form $M,2M,4M,8M,\ldots$ to $S$.
We define a minimum solution as follows.
\begin{definition}[Minimum solution]
For an ESS instance with input set $S$, a solution
$A, B \subseteq S$ is called a minimum solution if
$|A| + |B|$ is minimized.
We refer to the value $|A| + |B|$ as the minimum solution size,
denoted by $\ell$,
and define the minimum solution ratio as
$\ell' = \ell / n$.
\end{definition}

By trying all possible values $\ell \in [n]$, we incur only a polynomial overhead.
If no execution finds a solution, then the algorithm reports that the instance is a No instance.
Thus, in the following, we may assume that the minimum solution size $\ell$ is given in advance.

When analyzing the bounds in terms of the minimum solution ratio $\ell'$,
we may upper-bound the running time and space usage by maximizing over
the continuous interval $(0,1]$, instead of the discrete set
$\{1/n,2/n,\ldots,(n-1)/n,1\}$ of possible values of $\ell'$.

\section{Faster Exponential-Space Algorithm} \label{sec:FasterExpSpaceAlgorithms}

\begin{algorithm}[htbp]
\caption{RandomReplace}
\DontPrintSemicolon

\KwIn{Input set $S$ and replacement count $k$}
\KwOut{Modified set $S'$ and the set of new elements 
$E = \{e_1, \dots, e_{2k}\}$}

Randomly select a subset $D \subseteq S$ consisting of $2k$ elements,
and partition it arbitrarily into $k$ disjoint pairs.
Denote the $i$-th pair by $(x_i, y_i)$, where we relabel the elements
so that $x_i > y_i$ for each $i \in [k]$.\;

For each $i \in [k]$, define
$e_{2i-1} = x_i + y_i$,
and
$e_{2i} = x_i - y_i$.
Let $E = \{e_1, \ldots, e_{2k}\}$.\;

Construct the modified set
$S' = (S \setminus D) \cup E$.\;

\label{alg:RandomReplace}
\end{algorithm}

We first describe a modification of the input instance that may reduce
the number of integers required to represent a solution.
The algorithm randomly selects a subset of the input integers,
partitions them into disjoint pairs, and replaces each pair with its
sum and difference, as shown in \cref{alg:RandomReplace}.
In \cref{lm:Recover-Original-Solution}, we show that an ESS solution for
the modified instance can be transformed into a valid ESS solution for
the original instance under a certain condition.

\begin{lemma}\label{lm:Recover-Original-Solution}
Let $S'$ be the modified instance produced by
\cref{alg:RandomReplace}, and let $(x_i, y_i)$ denote the $i$-th pair
of elements obtained from $S$ by the algorithm.
Suppose that $S'$ admits an ESS solution
$A', B' \subseteq S'$ with $A' \cap B' = \emptyset$.
If, for every obtained pair $(x_i, y_i)$,
at most one of the generated values
$x_i+y_i$ and $x_i-y_i$ belongs to $A' \cup B'$,
then the solution $(A', B')$ for $S'$ can be transformed in
polynomial time into a valid ESS solution for the original instance $S$.
\end{lemma}

\begin{proof}
For each obtained pair $(x_i, y_i)$, we reconstruct the original
elements as follows.
If $x_i+y_i$ belongs to the solution, then we replace it by
$x_i$ and $y_i$ in the same subset.
If $x_i-y_i$ belongs to the solution, then we replace it by
$x_i$ in the same subset and place $y_i$ in the other subset.
Observe that these transformations preserve the equality of the subset
sums.
Applying this transformation independently to every obtained pair
yields a valid ESS solution for the original instance $S$.
Since the number of obtained pairs is at most $n/2$ and each pair is
processed using only a constant number of arithmetic operations, the
reconstruction procedure runs in polynomial time.
\end{proof}

By \cref{lm:preprocess-inputs}, all input integers are positive, and
thus $x_i+y_i > x_i-y_i$ for every obtained pair $(x_i, y_i)$.
During this modification process, if a generated value coincides with
another element of the modified instance, then these equal values
cannot originate from the same pair.
In this case, the corresponding elements immediately yield a valid ESS
solution for the original instance by
\cref{lm:Recover-Original-Solution}, and thus the algorithm may
terminate early.
Therefore, in the following, we assume that no value collisions occur
in the modified instance, and hence $|S| = |S'|$.

Moreover, the input assumption imposed by
\cref{lm:preprocess-inputs} continues to hold, since all generated
integers remain positive and are at most twice the maximum value in
the original input set $S$.

In \cref{lm:Obtain-Smaller-Solution}, we analyze the probability that
\cref{alg:RandomReplace} produces a modified instance with a solution of size $\ell-k$ that can be transformed into a valid solution for the original instance.

\begin{lemma}\label{lm:Obtain-Smaller-Solution}
Let $(S',E)$ be the output of \cref{alg:RandomReplace} applied to the input set $S$ with replacement count $k$, 
where $E=\{e_1,\dots,e_{2k}\}$.
Let $\ell$ denote the minimum solution size, and suppose that $2k \le \ell$.
Then, with probability at least $\binom{\ell}{2k} / \binom{n}{2k}$, 
there exists an ESS solution
$(A', B')$ for the modified instance $S'$ such that
$A' \cap B' = \emptyset$,
$|A'| + |B'| = \ell - k$, and
for every $i \in [k]$, the set $A' \cup B'$ contains exactly one of
$e_{2i-1}$ and $e_{2i}$.
Moreover, this solution can be transformed into a valid ESS solution
for the original instance $S$ by
\cref{lm:Recover-Original-Solution}.
\end{lemma}

\begin{proof}
Let $(A_{\min},B_{\min})$ be a minimum ESS solution for $S$ with
$|A_{\min}|+|B_{\min}|=\ell$.
Let $D \subseteq S$ be the random subset of size $2k$ chosen by
\cref{alg:RandomReplace}.
Consider the event that
$D \subseteq A_{\min}\cup B_{\min}$.
Since $A_{\min} \cap B_{\min} = \emptyset$, and $|A_{\min}\cup B_{\min}|=\ell$, this event occurs with probability
$\binom{\ell}{2k} / \binom{n}{2k}$.
We show that, conditioned on this event, the desired solution exists.

Assume that this event occurs.
For every $i \in [k]$, we have
$x_i, y_i \in A_{\min}\cup B_{\min}$,
where $(x_i,y_i)$ denotes the $i$-th pair obtained by partitioning $D$
in \cref{alg:RandomReplace}.
Hence, by applying the reverse transformation of
\cref{lm:Recover-Original-Solution},
the contribution of the pair $(x_i,y_i)$ to the solution can be encoded using exactly one of
$e_{2i-1}=x_i+y_i$ and $e_{2i}=x_i-y_i$.
Applying this transformation independently to all $k$ pairs yields an
ESS solution $(A',B')$ for $S'$ such that
$A'\cap B'=\emptyset$ and
$|A'|+|B'|=\ell-k$.
Moreover, for every $i\in[k]$, the set $A'\cup B'$ contains exactly one
of $e_{2i-1}$ and $e_{2i}$.
\end{proof}

\begin{algorithm}[htbp]
\caption{RandomReplacedBalancedESS}
\DontPrintSemicolon

\KwIn{Input set $S$, minimum solution size $\ell$, replacement count $k$, and prime bound $p_{\max}$}
\KwOut{An ESS solution for $S$ if one exists; otherwise, NO}

Apply \cref{alg:RandomReplace} to the input set $S$ with replacement count $k$ to obtain a modified input set $S'$ and the set of new elements $E = \{e_1, \ldots, e_{2k}\}$.\;

Pick a random prime $p$ in $[p_{\max}, 2p_{\max}]$, and a random residue $r$ in $[0, p-1]$.\;

Let $C=\{X\subseteq S' \mid \sum(X)\equiv_p r\}$.
Enumerate the elements of $C$ and store them one by one.
If the number of stored elements exceeds
$n^d \cdot \frac{2^n}{p_{\max}}$ for some fixed constant $d>2$,
halt and return \textnormal{NO}.

For each $X\in C$, let $v_X\in\{0,1\}^k$ be the vector defined by
\[
    (v_X)_i = |X\cap\{e_{2i-1},e_{2i}\}|\bmod 2
    \quad\text{for every } i\in[k].
\]
Let $C'=\{(X,v_X)\mid X\in C\}$.
Enumerate and store all elements of $C'$.\;

Find $(X,v_X),(Y,v_Y)\in C'$ such that
$\sum(X)=\sum(Y)$ and $v_X+v_Y=\mathbf{1}$, where vector addition is
componentwise and $\mathbf{1}$ denotes the all-one vector of length $k$;
if no such pairs exist, return NO.\;

Apply \cref{lm:Recover-Original-Solution} to the ESS solution $(X \setminus Y, Y \setminus X)$ for $S'$ to recover an ESS solution for the original input set $S$.

\label{alg:FastExpSpace}
\end{algorithm}

We then combine \cref{alg:RandomReplace} with the
\emph{BalancedEqualSubsetSum} algorithm of
\cite{mucha_et_al:LIPIcs.ESA.2019.73}.
Unlike the original algorithm, our modified algorithm first applies
\cref{alg:RandomReplace} to modify the input set.
It also constructs the additional set $C'$ in Line~4 of
\cref{alg:FastExpSpace}.
Even if $\sum(X)=\sum(Y)$ holds, it may be impossible to recover a solution for the original input set $S$ unless the condition in \cref{lm:Recover-Original-Solution} is satisfied.
We construct $C'$ in order to find a recoverable solution $A', B' \subseteq S'$
satisfying the condition guaranteed by \cref{lm:Obtain-Smaller-Solution};
namely, for every $i \in [k]$, the set $A' \cup B'$ contains exactly one of
$e_{2i-1}$ and $e_{2i}$. Since the algorithm outputs
$A' = X \setminus Y$ and $B' = Y \setminus X$, we have
$A' \cup B' = X \Delta Y$. Thus, it suffices to ensure that, for every
$i \in [k]$,
\[
\left(|X\cap\{e_{2i-1},e_{2i}\}|\bmod 2\right)
+
\left(|Y\cap\{e_{2i-1},e_{2i}\}|\bmod 2\right)
=1.
\]
The vector $v_X \in \{0,1\}^k$ stored in $C'$ records these parities for all
pairs $\{e_{2i-1},e_{2i}\}$. Hence, the constraint
$v_X+v_Y=\mathbf{1}$ enforces the above condition for every $i\in[k]$.
In \cref{lm:Size-of-Phi,lm:FastAlgSuccessProb}, we lower bound the probability that \cref{alg:FastExpSpace} outputs a solution.

\begin{lemma} \label{lm:Size-of-Phi}
Assume that the solution described in \cref{lm:Obtain-Smaller-Solution} exists, and let
\begin{equation} \label{eq:Phi-Definition}
\Phi =
\left\{
\sum(X) \;\middle|\;
\begin{aligned}
&X \subseteq S',\ \exists Y \subseteq S' \text{ such that }
X \neq Y,\ \sum(X)=\sum(Y),\\
&X \Delta Y \text{ contains exactly one of } e_{2i-1}
\text{ and } e_{2i}
\text{ for every } i \in [k]
\end{aligned}
\right\}.
\end{equation}
Then we have
\[
|\Phi| \ge
\begin{cases}
2^{n - (\ell - k)} & \text{if } \ell - k > \frac{n}{2}, \\
\binom{n - (\ell - k)}{\lfloor \frac{\ell - k - 1}{2}\rfloor} & \text{otherwise}.
\end{cases}
\]
\end{lemma}
\begin{proof}
Let $A', B' \subseteq S'$ be the solution described in
\cref{lm:Obtain-Smaller-Solution}, and let $Z = S' \setminus (A' \cup B')$.
For any subset $Z' \subseteq Z$, define $X = A' \cup Z'$ and
$Y = B' \cup Z'$. Then $X \neq Y$ and
$\sum(X) = \sum(A') + \sum(Z') = \sum(B') + \sum(Z') = \sum(Y)$.
Moreover, since the same set $Z'$ is added to both $A'$ and $B'$, we have
$X \Delta Y = A' \cup B'$. By the property of the solution in
\cref{lm:Obtain-Smaller-Solution}, $A' \cup B'$ contains exactly one of
$e_{2i-1}$ and $e_{2i}$ for every $i \in [k]$. Hence, $X$ and $Y$ satisfy
the defining conditions of $\Phi$.

Therefore, every distinct value of $\sum(Z')$ with $Z' \subseteq Z$ gives rise
to a distinct element of $\Phi$, shifted by the fixed value $\sum(A')$.
Consequently,
$|\Phi| \ge \left|\left\{\sum(Z') \mid Z' \subseteq Z\right\}\right|$.
It remains to lower bound the number of distinct subset sums of $Z$.

First, consider the case $\ell - k > n/2$. We claim that all subset sums of $Z$ are
distinct. Suppose, for contradiction, that there exist distinct subsets
$Z_1,Z_2 \subseteq Z$ such that $\sum(Z_1)=\sum(Z_2)$. Then
$(Z_1 \setminus Z_2, Z_2 \setminus Z_1)$ is an ESS solution for $S'$.
Since $A' \cup B'$ contains exactly one of $e_{2i-1}$ and $e_{2i}$ for every
$i \in [k]$, the set $Z$ also contains exactly one of $e_{2i-1}$ and
$e_{2i}$ for every $i \in [k]$. Thus, by
\cref{lm:Recover-Original-Solution}, this ESS solution for $S'$ can be
converted into an ESS solution for the original set $S$ by increasing its size
by at most $k$. Its size is at most $|Z|+k = n-\ell+2k$, which is strictly
smaller than $\ell$ because $\ell-k > n/2$. This contradicts the assumption
that the minimum solution size for $S$ is $\ell$.
Therefore, all subset sums of $Z$ are distinct, and hence
$\left|\left\{\sum(Z') \mid Z' \subseteq Z\right\}\right| = 2^{|Z|}
= 2^{n-(\ell-k)}$, which gives the desired lower bound in this case.

Second, consider the remaining case. Let
$j=\left\lfloor \frac{\ell-k-1}{2} \right\rfloor$. We claim that all subsets of $Z$ of size $j$ have
distinct sums. Suppose, for contradiction, that there exist distinct subsets
$Z_1,Z_2 \subseteq Z$ such that $|Z_1|=|Z_2|=j$ and
$\sum(Z_1)=\sum(Z_2)$. Then
$(Z_1 \setminus Z_2, Z_2 \setminus Z_1)$ is an ESS solution for $S'$.
Again, by \cref{lm:Recover-Original-Solution}, this solution can be converted
into an ESS solution for $S$ by increasing its size by at most $k$. The
resulting solution has size at most $2j+k < \ell$, contradicting the
minimality of $\ell$. Therefore, all subsets of $Z$ of size $j$ have distinct
sums, and hence $\left|\left\{\sum(Z') \mid Z' \subseteq Z\right\}\right| \ge \binom{|Z|}{j}$,
which gives the desired lower bound.
\end{proof}

\begin{lemma} \label{lm:FastAlgSuccessProb}
Assume that $0 \le 2k \le \ell$ and
$\epsilon n \le \ell \le (1-\epsilon)n$ for some constant $\epsilon > 0$
(say, $\epsilon = 1/100$).
Set the prime-bound parameter $p_{\max}$ in \cref{alg:FastExpSpace} to
$2^{n-(\ell-k)}$ if $\ell-k > n/2$, and to
$\binom{n-(\ell-k)}{\lfloor(\ell-k-1)/2\rfloor}$ otherwise.
Then, with probability at least
$\Omega^*\left(\binom{\ell}{2k}/\binom{n}{2k}\right)$,
\cref{alg:FastExpSpace} outputs a solution.
\end{lemma}

\begin{proof}
Let $\mathcal{F}$ denote the event that the solution described in
\cref{lm:Obtain-Smaller-Solution} exists.
We analyze the conditional probability of success given that
$\mathcal{F}$ occurs.
By \cref{lm:Obtain-Smaller-Solution}, we have
$\mathbb{P}[\mathcal{F}]
\ge \binom{\ell}{2k}/\binom{n}{2k}$.
It remains to show that, conditioned on $\mathcal{F}$,
\cref{alg:FastExpSpace} outputs a solution with probability $\Omega^*(1)$.

Let $\mathcal{G}$ denote the event that \cref{alg:FastExpSpace} does not
halt in Line~3, i.e., the event that
$|C|\le n^d \frac{2^n}{p_{\max}}$ occurs, where $d>2$ is the fixed
constant used in Line~3 of \cref{alg:FastExpSpace}.
When $\mathcal{G}$ occurs, the algorithm reaches Line~5.

Let $\mathcal{H}$ denote the event that there exist
$(X,v_X),(Y,v_Y)\in C'$ satisfying
$\sum(X)=\sum(Y)$ and $v_X+v_Y=\mathbf{1}$.
When $\mathcal{H}$ occurs, the algorithm finds such a pair in Line~5.

Therefore, the conditional probability that \cref{alg:FastExpSpace}
outputs a solution can be expressed as
$\mathbb{P}[(\mathcal{G}\cap\mathcal{H})\mid\mathcal{F}]$.
To lower-bound $\mathbb{P}[\mathcal{H}\mid\mathcal{F}]$, it suffices to
lower-bound the probability that there exists an $a\in\Phi$ such that
$a\equiv_p r$, where $\Phi$ is defined in \eqref{eq:Phi-Definition} and
$r$ is the random residue chosen in Line~2 of \cref{alg:FastExpSpace}.

\begin{restatable}{claim}{lowerboundphimodp}
\label{clm:Lower-Bound-Phi-mod-p}
Suppose that $|\Phi|\ge p_{\max}$ and $p_{\max}=2^{\Theta(n)}$.
Then, with constant probability, we have
$|\{a \bmod p \mid a\in\Phi\}| \ge \Omega(p/n^2)$.
\end{restatable}

\begin{proof}
The proof is identical to that of
\cite[Lemma~3.6]{mucha_et_al:LIPIcs.ESA.2019.73}.
For completeness, we provide the details in
Appendix~\ref{app:Repeat-Proof-Of-ESA-Lemma3.6}.
\end{proof}

By \cref{lm:Size-of-Phi}, conditioned on $\mathcal{F}$, we have
$|\Phi|\ge p_{\max}$.
Moreover, from the assumptions on $\ell$ and $k$, we have
$\ell-k=\Theta(n)$ and $n-(\ell-k)=\Theta(n)$.
Thus, in either case of the definition of $p_{\max}$, we have
$p_{\max}=2^{\Theta(n)}$.
Hence, by \cref{clm:Lower-Bound-Phi-mod-p}, with constant probability,
$|\{a\bmod p\mid a\in\Phi\}|\ge \Omega(p/n^2)$.
Conditioned on this event, since $r$ is chosen uniformly at random from
$[0,p-1]$, the probability that
$r\in\{a\bmod p\mid a\in\Phi\}$ is $\Omega(1/n^2)$.
Therefore,
\begin{equation} \label{eq:H-and-F-Prob}
    \mathbb{P}[\mathcal{H}\mid\mathcal{F}] \ge \Omega(1/n^2).
\end{equation}
We next upper-bound the expected value of $|C|$.

\begin{claim} \label{clm:Exp-Size-Of-C} For any fixed $S'$, we have $\mathbb{E}[|C|]\le 2^n/p_{\max}$. \end{claim}

\begin{proof}
For a fixed prime $p$, each subset $X\subseteq S'$ falls into $C$ with probability $1/p$, since the residue $r$ is chosen uniformly at random from $[0,p-1]$.
Since $p\ge p_{\max}$, this probability is at most $1/p_{\max}$ for any $p$.
Thus, the result follows by linearity of expectation.
\end{proof}

Since \cref{clm:Exp-Size-Of-C} holds for every fixed outcome of $S'$, it
also holds for the conditional expectation given $\mathcal{F}$.
Hence, $\mathbb{E}[|C|\mid\mathcal{F}]\le 2^n/p_{\max}$.
Consider the complement of $\mathcal{G}$, namely the event
$\mathcal{G}^c$ that \cref{alg:FastExpSpace} halts in Line~3, i.e., the
event that $|C|>n^d \frac{2^n}{p_{\max}}$ occurs.
By Markov's inequality, we have
\begin{equation} \label{eq:G^c-and-F-prob}
\mathbb{P}[\mathcal{G}^c\mid\mathcal{F}]
\le
\frac{\mathbb{E}[|C|\mid\mathcal{F}]}
{n^d \cdot 2^n/p_{\max}}
\le
\frac{1}{n^d}.
\end{equation}
We are now ready to prove
$\mathbb{P}[(\mathcal{G}\cap\mathcal{H})\mid\mathcal{F}]
\ge \Omega^*(1)$:
\begin{align}
\mathbb{P}[(\mathcal{G}\cap\mathcal{H})\mid\mathcal{F}]
&=
\mathbb{P}[\mathcal{H}\mid\mathcal{F}]
-
\mathbb{P}[(\mathcal{G}^c\cap\mathcal{H})\mid\mathcal{F}] \\
&\ge
\Omega(1/n^2)-1/n^d
\quad\text{by \eqref{eq:H-and-F-Prob} and \eqref{eq:G^c-and-F-prob}} \\
&=
\Omega(1/n^2) = \Omega^*(1),
\quad\text{since $d>2$.}
\end{align}
This completes the proof of \cref{lm:FastAlgSuccessProb}.
\end{proof}

We now turn to the time and space complexity analysis.

\begin{lemma} \label{lm:Enum-C-time-space}
The set $C$ can be enumerated in
$O^*(\max\{|C|, 2^{n/2}\})$ time and
$O^*(\max\{|C|, 2^{n/4}\})$ space.
\end{lemma}

\begin{proof}
This is the modular variant of the Subset-Sum problem, and the technique of
Schroeppel and Shamir~\cite{SchroeppelAndShamir1981} can compute it in
$O^*(2^{n/2})$ time and $O^*(2^{n/4})$ space, excluding the output size.
See \cite[Section~3.2]{becker2011improved} for pseudocode of the modular version.
The pseudocode requires scanning the subset sums of the two halves of $S'$
in ascending and descending order with respect to their values modulo $p$.
A straightforward implementation would enumerate and sort all $O^*(2^{n/2})$
sums, using $O^*(2^{n/2})$ space. In Appendix \ref{app:Modulo-Schroeppel-Shamir}, we explain that a minor
modification of the technique of Schroeppel and
Shamir~\cite{SchroeppelAndShamir1981} performs the same scan using only
$O^*(2^{n/4})$ space.
\end{proof}

Since we enumerate all elements of $C'$, Line 5 of \cref{alg:FastExpSpace}
can be implemented in $O^*(|C|)$ time and space. 
Specifically, we sort the elements of $C'$ by the key
$(\sum(X),v_X)$. Then, for each group of elements with the same value of
$\sum(X)$, we scan the corresponding vectors $v_X$ and search for the
complementary vector $\mathbf{1}-v_X$ by binary search. If such
a vector is found in the same group, then we obtain a pair
$(X,v_X),(Y,v_Y)\in C'$ such that $\sum(X)=\sum(Y)$ and
$v_X+v_Y=\mathbf{1}$.

Both \cref{alg:RandomReplace} and the recovery of the ESS solution for the
original input set $S$ require only polynomial time. Therefore, the
bottleneck in \cref{alg:FastExpSpace} is the enumeration of the set $C$.
Choosing the prime-bound parameter as in \cref{lm:FastAlgSuccessProb} yields \cref{thm:FastExpSpaceComplexity}.

\begin{theorem}
\label{thm:FastExpSpaceComplexity}
Let $\ell$ be the minimum solution size, where
$\epsilon n \le \ell \le (1-\epsilon)n$ for some constant $\epsilon>0$,
and let $\ell'=\ell/n$.
Let $k'$ be a real parameter satisfying $0\le k'\le \ell'/2$, and set
the replacement count to $k=\lfloor k'n\rfloor$.
Define
$\mathrm{repl}(\ell',k')=H(2k')-\ell'H(2k'/\ell')$
and
\begin{equation} \label{eq:phi(x)}
\mathrm{phi}(x)=(1-x)H\left(\frac{x}{2(1-x)}\right).
\end{equation}
Then there is a Monte Carlo algorithm for ESS with time complexity
$O^*(2^{\tau_{\mathrm{rrb}}(\ell',k')n})$ and space complexity
$O^*(2^{\sigma_{\mathrm{rrb}}(\ell',k')n})$, where
\begin{equation}
\label{eq:tau-rrb}
\tau_{\mathrm{rrb}}(\ell',k') =
\begin{cases}
\ell'-k'+\mathrm{repl}(\ell',k')
& \text{if } \ell'-k'>1/2,\\[1mm]
\max\{1-\mathrm{phi}(\ell'-k'),1/2\}
+\mathrm{repl}(\ell',k')
& \text{otherwise.}
\end{cases}
\end{equation}
and
\begin{equation}
\label{eq:sigma-rrb}
\sigma_{\mathrm{rrb}}(\ell',k') =
\begin{cases}
\ell'-k'
& \text{if } \ell'-k'>1/2,\\[1mm]
1-\mathrm{phi}(\ell'-k')
& \text{otherwise.}
\end{cases}
\end{equation}
\end{theorem}

\begin{proof}
Let $x=\ell'-k'$.
As in \cref{lm:FastAlgSuccessProb}, set
$p_{\max}=2^{n-(\ell-k)}=\Theta^*(2^{(1-x)n})$ if $x>1/2$, and
$p_{\max}=\binom{n-(\ell-k)}{\lfloor(\ell-k-1)/2\rfloor}
=\Theta^*(2^{\mathrm{phi}(x)n})$ otherwise.
The latter equality follows from
$\left|\frac{\lfloor(\ell-k-1)/2\rfloor}{n-(\ell-k)}
-\frac{x}{2(1-x)}\right| \le O(1/n)$, together with
\cref{fact:H(x+1/n)-H(x)}.
Here, we use the assumptions on $\ell$ and $k'$, which imply
$\ell-k=\Theta(n)$ and $n-(\ell-k)=\Theta(n)$.

By \cref{lm:Enum-C-time-space}, one trial of \cref{alg:FastExpSpace} runs in
$O^*(\max\{2^{xn},2^{n/2}\})=O^*(2^{xn})$ time and uses
$O^*(\max\{2^{xn},2^{n/4}\})=O^*(2^{xn})$ space when $x>1/2$.
In the remaining case, one trial runs in
$O^*(\max\{2^{(1-\mathrm{phi}(x))n},2^{n/2}\})$ time and uses
$O^*(\max\{2^{(1-\mathrm{phi}(x))n},2^{n/4}\})
=O^*(2^{(1-\mathrm{phi}(x))n})$ space.
By \cref{lm:FastAlgSuccessProb}, one trial succeeds with probability at least
$\Omega^*\left(\binom{\ell}{2k}/\binom{n}{2k}\right)$.
Hence, repeating the algorithm
$O^*(\binom{n}{2k}/\binom{\ell}{2k})
=O^*(2^{\mathrm{repl}(\ell',k')n})$ times amplifies the success
probability to a constant, where the bound follows from
$\left|2k/n-2k'\right|\le O(1/n)$ and
$\left|2k/\ell-2k'/\ell'\right| \le O(1/\ell)$, together with
\cref{fact:H(x+1/n)-H(x)}.

Therefore, when $x>1/2$, the total running time is
$O^*(2^{(x+\mathrm{repl}(\ell',k'))n})$, and the space usage is
$O^*(2^{xn})$.
When $x\le 1/2$, the total running time is
$O^*\left(2^{(\max\{1-\mathrm{phi}(x),1/2\}
+\mathrm{repl}(\ell',k'))n}\right)$,
and the space usage is
$O^*(2^{(1-\mathrm{phi}(x))n})$.
Since $x=\ell'-k'$, these bounds are exactly
$O^*(2^{\tau_{\mathrm{rrb}}(\ell',k')n})$ time and
$O^*(2^{\sigma_{\mathrm{rrb}}(\ell',k')n})$ space, as claimed.
\end{proof}

In the remainder of this section, we mainly focus on the case
$\ell'-k'>1/2$. The case $\ell'-k'\le 1/2$ will be used later when deriving
the time-space tradeoff in \cref{sec:time-space-tradeoff}.
The authors of \cite{mucha_et_al:LIPIcs.ESA.2019.73} presented the
\emph{UnbalancedEqualSubsetSum} algorithm, which is particularly useful when
$\ell'$ is close to $0$ or $1$.

\begin{lemma}[{\cite[Theorem 3.3]{mucha_et_al:LIPIcs.ESA.2019.73}}] \label{lm:UnbalancedESS}
Let $\ell' = \ell / n$ be the minimum solution size ratio.
There is a Monte Carlo algorithm for ESS that runs in
$O^*(2^{\frac{n}{2}(H(\ell')+\ell')})$ time and uses
$O^*(2^{\frac{n}{4}(H(\ell')+\ell')})$
space.
\end{lemma}

The authors of \cite{mucha_et_al:LIPIcs.ESA.2019.73} actually give an algorithm
using $O^*(2^{\frac{n}{2}(H(\ell')+\ell')})$ space. This space usage can be
reduced to $O^*(2^{\frac{n}{4}(H(\ell')+\ell')})$ by applying the technique of
Schroeppel and Shamir~\cite{SchroeppelAndShamir1981}; we provide the details
in Appendix~\ref{app:SmallSpaceUnbalancedESS}.
We now restate and prove \cref{thm:MainResult-ExpSpace}.

\FastExpSpaceThm*

\begin{proof}
Let $\gamma_1$ be the value of $\ell'$ in the interval
$(1/\sqrt{2},1-\epsilon]$ satisfying
\begin{equation} \label{eq:gamma_1}
\frac{H(\ell')+\ell'}{2}
=
\tau_{\mathrm{rrb}}\left(
    \ell',
    \frac{\sqrt{2}\ell'-1}{2(\sqrt{2}-1)}
\right).
\end{equation}
We define $\alpha_1$ by
\begin{equation} \label{eq:alpha_1}
    \alpha_1
    =
    \gamma_1-\frac{\sqrt{2}\gamma_1-1}{2(\sqrt{2}-1)}.
\end{equation}
Numerically, these values satisfy
$0.79158<\gamma_1<0.79159$ and $0.64737<\alpha_1<0.64738$.

For $\ell'\in(0,1/2]\cup(\gamma_1,1]$, we apply
\cref{lm:UnbalancedESS}.
For $\ell'\in(1/2,\gamma_1]$, we apply
\cref{thm:FastExpSpaceComplexity} with
$k'=\max\{0,\ell'-\alpha_1\}$ for each value of $\ell'$.
Then the worst-case time exponent, namely the constant $c$ such that the
running time is $O^*(2^{cn})$, can be expressed as
\[
\max\left\{
\max_{\ell'\in(0,1/2]\cup(\gamma_1,1]}
    \frac{H(\ell')+\ell'}{2},
\quad
\max_{\ell'\in(1/2,\gamma_1]}
    \tau_{\mathrm{rrb}}\left(
        \ell',
        \max\{0,\ell'-\alpha_1\}
    \right)
\right\},
\]
and the worst-case space exponent, defined analogously, can be expressed as
\[
\max\left\{
\max_{\ell'\in(0,1/2]\cup(\gamma_1,1]}
    \frac{H(\ell')+\ell'}{4},
\quad
\max_{\ell'\in(1/2,\gamma_1]}
    \sigma_{\mathrm{rrb}}\left(
        \ell',
        \max\{0,\ell'-\alpha_1\}
    \right)
\right\}.
\]
Both are maximized at $\ell'=\gamma_1$.
At this point, the time exponent is upper-bounded by
$\frac{H(\gamma_1)+\gamma_1}{2}
=
\tau_{\mathrm{rrb}}\left(
    \gamma_1,
    \gamma_1-\alpha_1
\right)
\le 0.76502$,
and the space exponent is upper-bounded by
$\sigma_{\mathrm{rrb}}\left(
    \gamma_1,
    \gamma_1-\alpha_1
\right)
= \alpha_1 \le 0.64738$.
Therefore, the worst-case time complexity is
$O^*(2^{0.76502n})\le O^*(1.6994^n)$, and the worst-case space complexity is
$O^*(2^{0.64738n})\le O^*(1.5664^n)$.
\end{proof}

\section{Faster Polynomial-Space Algorithm} \label{sec:FasterPolySpaceAlgorithm}
For algorithms restricted to $\mathrm{poly}(n)$ space, we can employ the
low-space element distinctness algorithm. This algorithm was introduced in
\cite{FOCS2013-ElementDistinctness}, generalized in
\cite{Bansal2018-FastPolySpaceSubsetSum}, and later freed from non-standard
assumptions in \cite{chen2022truly,lyu2023time}.
\cref{lm:LowSpaceElementDistinct-ManySolutions} can be inferred from
\cite[Section~4.2, proof of Theorem~1.1]{chen2022truly}.
The same extracted lemma is also stated in
\cite[Theorem~12]{jin_et_al:LIPIcs.ICALP.2024.94}.

\begin{lemma}[\cite{chen2022truly}]
\label{lm:LowSpaceElementDistinct-ManySolutions}
Given random access to an array of $N$ nonnegative integers
$a_1,\ldots,a_N$, with $a_i \le \mathrm{poly}(N)$ for all $i\in[N]$,
suppose that the array contains at least one colliding pair, that is, a
pair $(i,j)$ with $i\neq j$ and $a_i=a_j$.
Then there is a Monte Carlo algorithm that reports such a pair using only
$\mathrm{polylog}(N)$ space and
\[
  O\left(
    \frac{N\sqrt{F_2}}{F_2-N}\,\mathrm{polylog}(N)
  \right)
\]
time, where
$F_2=\sum_{i=1}^N\sum_{j=1}^N \mathbf{1}[a_i=a_j]$
is the second frequency moment.
In particular, under the above assumption, $F_2\in[N+2,N^2]$.
\end{lemma}

Although the statement in \cite{chen2022truly} is phrased for positive
integers, the same bound applies to nonnegative integers by adding $1$ to
every array entry, which preserves all collisions.
We apply \cref{lm:LowSpaceElementDistinct-ManySolutions} to the $2^n$
subset sums of $S$ to obtain \cref{thm:(0.5+l')-time}.

\begin{theorem}
\label{thm:(0.5+l')-time}
Let $\ell'$ be the minimum solution ratio. Then there is a Monte Carlo
algorithm for ESS that runs in $O^*(2^{(\ell'+1/2)n})$ time using
polynomial space.
\end{theorem}

\begin{proof}
We provide random access to the list $L$ of all subset sums of
$S=\{w_1,\ldots,w_n\}$ by computing each element on the fly. Specifically,
for each index $i\in[N]$, where $N=2^n$, let
\[
X_i = \{j\in[n] \mid \text{the } j\text{-th least significant bit of the binary representation of } i-1 \text{ is } 1\}.
\]
Then the $i$-th element of $L$ is computed as
$L_i = \sum_{j\in X_i} w_j$.
This gives random access to $L$ using $\mathrm{polylog}(N)$ time and space.
We apply \cref{lm:LowSpaceElementDistinct-ManySolutions} to $L$. A
colliding pair in $L$ corresponds exactly to two distinct
subsets of $S$ with equal sum, and hence yields an ESS solution.

It remains to bound the running time. Let $A,B\subseteq S$ be a minimum
solution, so that $|A\cup B|=\ell'n$. Let
$Z=S\setminus(A\cup B)$; then $|Z|=(1-\ell')n$. For every subset
$Z'\subseteq Z$, the two subsets $A\cup Z'$ and $B\cup Z'$ have equal
sum. Therefore, these pairs contribute at least
$2^{|Z|}=2^{(1-\ell')n}$ nontrivial terms to $F_2$, and hence
$F_2 \ge N + 2^{(1-\ell')n}$.
Since the function $\sqrt{x}/(x-N)$ is monotonically decreasing for
$x\in[N+2,N^2]$, we can upper-bound the running time by setting
$F_2=N+2^{(1-\ell')n}$. Thus,
\[
  O\left(\frac{N\sqrt{F_2}}{F_2-N}\,\mathrm{polylog}(N)\right)
  \le
  O\left(
    \frac{2^n\sqrt{2^n+2^{(1-\ell')n}}}{2^{(1-\ell')n}}
    \,\mathrm{poly}(n)
  \right)
  \le
  O^*(2^{(\ell'+1/2)n}).
\]
Hence, the algorithm runs in $O^*(2^{(\ell'+1/2)n})$ time and uses
polynomial space.
\end{proof}

A fast polynomial-space algorithm for the Subset-Sum problem was given in
\cite{Bansal2018-FastPolySpaceSubsetSum}, and the non-standard assumptions
used in its analysis were later removed in \cite{chen2022truly}.

\begin{lemma}[{\cite{Bansal2018-FastPolySpaceSubsetSum,chen2022truly}}]
\label{lm:Bansal-Fast-PolySpace-SubsetSum}
There is a Monte Carlo algorithm for Subset Sum that runs in
$O^*(2^{0.86n})$ time using polynomial space.
\end{lemma}

We use the polynomial-space algorithm for the Subset Sum problem as a
subroutine to prove \cref{thm:Use-SubsetSum-As-Routine}.

\begin{theorem}
\label{thm:Use-SubsetSum-As-Routine}
Let $\ell'$ be the minimum solution ratio. Then there is a Monte Carlo
algorithm for ESS that runs in
$O^*(2^{(H(\ell')+0.86\ell')n})$ time using polynomial space.
\end{theorem}

\begin{proof}
Let $\ell=\ell'n$ be the minimum solution size. We iterate over all subsets
$X\subseteq S$ of size $\ell$, keeping only the current subset in memory.
For each such subset $X$, if $\sum(X)$ is odd, then $X$ cannot be
partitioned into two subsets of equal sum, so we skip this case. Otherwise,
we run the polynomial-space Subset-Sum algorithm of
\cref{lm:Bansal-Fast-PolySpace-SubsetSum} on the instance $X$ with target
$\sum(X)/2$.

Since there exists a minimum ESS solution $A,B\subseteq S$ such that
$|A|+|B|=\ell$ and $A\cap B=\emptyset$, the iteration eventually considers
$X=A\cup B$, and the corresponding Subset-Sum instance with target
$\sum(X)/2$ has a solution. The running time is therefore
\[
  O^*\left(\binom{n}{\ell}2^{0.86\ell}\right)
  \le
  O^*(2^{(H(\ell')+0.86\ell')n}).
\]
The space usage remains polynomial, since we iterate over the subsets $X$
one by one while storing only the current subset, and each call to the
Subset-Sum algorithm uses polynomial space.
\end{proof}

We combine \cref{thm:(0.5+l')-time,thm:Use-SubsetSum-As-Routine} to prove
\cref{thm:MainResult-PolySpace}, which we restate here.

\FastPolySpaceThm*

\begin{proof}
Let $\gamma_2$ be the larger value of $\ell'$ at which the functions
$\ell'+1/2$ and $H(\ell')+0.86\ell'$ intersect, namely the solution to
\begin{equation}
\label{eq:gamma_2}
    \ell' + \frac{1}{2}
    =
    H(\ell') + 0.86\ell'
    \quad\text{for } \frac{1}{2}<\ell'<1.
\end{equation}
Numerically, $0.8465<\gamma_2<0.84651$.

For $\ell'\in(0,\gamma_2]$, we apply \cref{thm:(0.5+l')-time}.
For $\ell'\in(\gamma_2,1]$, we apply
\cref{thm:Use-SubsetSum-As-Routine}.
Since \cref{thm:(0.5+l')-time,thm:Use-SubsetSum-As-Routine} both solve
ESS in polynomial space, it remains to bound the running time.
Then the worst-case time exponent, namely the constant $c$ such that the running time is $O^*(2^{cn})$, can be expressed as
\[
\max\left\{
\max_{\ell'\in(0,\gamma_2]}
    \left\{\ell'+\frac{1}{2}\right\},
\quad
\max_{\ell'\in(\gamma_2,1]}
    \left\{H(\ell')+0.86\ell'\right\}
\right\}.
\]
This expression is maximized at $\ell'=\gamma_2$, and its value is at most
$\gamma_2+1/2=H(\gamma_2)+0.86\gamma_2\le 1.34651$.
Thus, the worst-case running time is
$O^*(2^{1.34651n})\le O^*(2.5430^n)$.
\end{proof}

\section{Time-Space Tradeoffs}
\label{sec:time-space-tradeoff}

We consider the setting where the space usage is bounded by
$O^*(2^{\alpha n})$ for a constant $\alpha$, and study the resulting time
complexity.
In \cref{sec:CaseDistinctionOverview}, we give an overview of the
tradeoffs by summarizing, in a table, which lemma or theorem is used for
each range of $\alpha$ and $\ell'$.
In \cref{sec:5-2,sec:5-3,sec:5-4,sec:5-5}, we explain the theorems used
in \cref{sec:CaseDistinctionOverview}.

\subsection{Overview of the tradeoffs} \label{sec:CaseDistinctionOverview}

\begin{table}[htbp]
    \centering
    \begin{tabular}{|C{0.21\textwidth}|C{0.6\textwidth}|C{0.09\textwidth}|}
        \hline
        Range of $\alpha$
        & Range of $\ell'$
        & Worst $\ell'$ \\
        \hline

        $1/2 < \alpha \le \alpha_1$
        &
        \begin{tabularx}{\linewidth}{@{}Y|Y@{}}
            $\ell' \in (0, 1/2] \cup (\ell'_{1}(\alpha), 1]$
            &
            $\ell' \in (1/2, \ell'_{1}(\alpha)]$
            \\
            \cref{lm:UnbalancedESS}
            &
            \cref{thm:RandomReplace-0.5-0.76-space}
        \end{tabularx}
        & $\ell'_{1}(\alpha)$ \\
        \hline

        $\alpha_2 < \alpha \le 1/2$
        &
        \begin{tabularx}{\linewidth}{@{}Y|Y@{}}
            $\ell' \in (0, 1/2] \cup (\ell'_{2}(\alpha), 1]$
            &
            $\ell' \in (1/2, \ell'_{2}(\alpha)]$
            \\
            \cref{lm:UnbalancedESS}
            &
            \cref{thm:RandomReplace-0.45-0.5-space}
        \end{tabularx}
        & $\ell'_{2}(\alpha)$ \\
        \hline

        $2-\log 3 < \alpha \le \alpha_2$
        &
        \begin{tabularx}{\linewidth}{@{}Y|Y@{}}
            $\ell' \in (0, 1/2] \cup (\ell'_{3}(\alpha), 1]$
            &
            $\ell' \in (1/2, \ell'_{3}(\alpha)]$
            \\
            \cref{lm:UnbalancedESS}
            &
            \cref{thm:SpaceBoundedBalancedESS-TimeComplexity}
        \end{tabularx}
        & $\ell'_{3}(\alpha)$ \\
        \hline

        $\frac{2}{3}(\log 3 - 1) < \alpha \le 2-\log 3$
        &
        \begin{tabularx}{\linewidth}{@{}Y@{}}
            $\ell' \in (0, 1]$ \\
            \cref{thm:SpaceBoundedUnbalancedESS}
        \end{tabularx}
        & $2/3$ \\
        \hline

        $\frac{1}{4} < \alpha \le \frac{2}{3}(\log 3 - 1)$
        &
        \begin{tabularx}{\linewidth}{@{}Y|Y@{}}
            $\ell' \in (0,1/2] \cup (\ell'_{4}(\alpha),1]$
            &
            $\ell' \in (1/2,\ell'_{4}(\alpha)]$
            \\
            \cref{thm:SpaceBoundedUnbalancedESS}
            &
            \cref{thm:SpaceBoundedBalancedESS-TimeComplexity}
        \end{tabularx}
        & $\ell'_{4}(\alpha)$ \\
        \hline

        $\alpha_3 < \alpha \le 1/4$
        &
        \begin{tabularx}{\linewidth}{@{}Y|Y|Y@{}}
            $\ell' \in (0,1-2\alpha]$
            &
            $\ell' \in (1-2\alpha,\ell'_{5}(\alpha)]$
            &
            $\ell' \in (\ell'_{5}(\alpha),1]$
            \\
            \cref{thm:(0.5+l')-time}
            &
            \cref{thm:SpaceBoundedBalancedESS-TimeComplexity}
            &
            \cref{thm:SpaceBoundedUnbalancedESS}
        \end{tabularx}
        & $\ell'_{5}(\alpha)$ \\
        \hline

        $\alpha_4 < \alpha \le \alpha_3$
        &
        \begin{tabularx}{\linewidth}{@{}Y|Y|Y@{}}
            $\ell' \in (0,1-2\alpha]$
            &
            $\ell' \in (1-2\alpha,\ell'_{6}(\alpha)]$
            &
            $\ell' \in (\ell'_{6}(\alpha),1]$
            \\
            \cref{thm:(0.5+l')-time}
            &
            \cref{thm:SpaceBoundedBalancedESS-TimeComplexity}
            &
            \cref{thm:SpaceBounded-UseSubset-ESS}
        \end{tabularx}
        & $\ell'_{6}(\alpha)$ \\
        \hline

        $0 \le \alpha \le \alpha_4$
        &
        \begin{tabularx}{\linewidth}{@{}Y|Y@{}}
            $\ell' \in (0,\ell'_{7}(\alpha)]$
            &
            $\ell' \in (\ell'_{7}(\alpha), 1]$
            \\
            \cref{thm:(0.5+l')-time}
            &
            \cref{thm:SpaceBounded-UseSubset-ESS}
        \end{tabularx}
        & $\ell'_{7}(\alpha)$ \\
        \hline
        
    \end{tabular}
    \caption{Case distinction by $\alpha$ and $\ell'$.}
    \label{tab:case-distinction}
\end{table}

We first give an overview of the case distinction in
\cref{tab:case-distinction}.
For each range of the space-bound parameter $\alpha$, the table indicates
which lemma or theorem is applied for each range of the minimum solution
ratio $\ell'$, together with the bottleneck value of $\ell'$ in that case.
By the time-complexity exponent function of a lemma or theorem, we mean
the function $f(\ell',\alpha)$ such that, under the space bound
$O^*(2^{\alpha n})$, the claimed time complexity is
$O^*(2^{f(\ell',\alpha)n})$.
The space complexity of \cref{lm:UnbalancedESS} is bounded by
$O^*(3^{n/4})$.
Thus, when $\frac{\log 3}{4}<2-\log 3\le\alpha$, we have
$O^*(3^{n/4})\le O^*(2^{\alpha n})$, and hence
\cref{lm:UnbalancedESS} satisfies the space bound $O^*(2^{\alpha n})$.
The other theorems used in \cref{tab:case-distinction} explicitly state
their time complexities under the space bound $O^*(2^{\alpha n})$.

Here, $\alpha_1$ is defined by \eqref{eq:alpha_1} and satisfies
$0.64737<\alpha_1<0.64738$.

The value $\alpha_2$ is defined as the value of
$\alpha\in[0.45,1/2]$ for which there exists
$\ell'\in(1/2,0.8)$ such that the three time-complexity exponent
functions from
\cref{lm:UnbalancedESS,thm:RandomReplace-0.45-0.5-space,thm:SpaceBoundedBalancedESS-TimeComplexity}
all take the same value at $\ell'$.
Numerically, we have $0.49211<\alpha_2<0.49212$.

Similarly, $\alpha_3$ is defined as the value of $\alpha<1/4$ for which
there exists $\ell'>1/2$ such that the three time-complexity exponent
functions from
\cref{thm:SpaceBoundedBalancedESS-TimeComplexity,thm:SpaceBoundedUnbalancedESS,thm:SpaceBounded-UseSubset-ESS}
all take the same value at $\ell'$.
Numerically, we have $0.18686<\alpha_3<0.18687$.

Finally, $\alpha_4$ is defined as the value of $\alpha<\alpha_3$ for which
there exists $\ell'$ such that the three time-complexity exponent
functions from
\cref{thm:SpaceBoundedBalancedESS-TimeComplexity,thm:(0.5+l')-time,thm:SpaceBounded-UseSubset-ESS}
all take the same value at $\ell'$.
Numerically, we have $0.12285<\alpha_4<0.12286$.

For each $i\in[7]$, the value $\ell'_i(\alpha)$ in
\cref{tab:case-distinction} denotes the value of $\ell'>1/2$ at which the
two corresponding time-complexity exponent functions intersect.
The corresponding pairs are as follows:
\begin{itemize}
    \item $\ell'_1(\alpha)$: \cref{lm:UnbalancedESS} and
    \cref{thm:RandomReplace-0.5-0.76-space}, for
    $1/2<\alpha\le\alpha_1$.
    \item $\ell'_2(\alpha)$: \cref{lm:UnbalancedESS} and
    \cref{thm:RandomReplace-0.45-0.5-space}, for
    $\alpha_2<\alpha\le1/2$.
    \item $\ell'_3(\alpha)$: \cref{lm:UnbalancedESS} and
    \cref{thm:SpaceBoundedBalancedESS-TimeComplexity}, for
    $2-\log 3<\alpha\le\alpha_2$.
    \item $\ell'_4(\alpha)$:
    \cref{thm:SpaceBoundedBalancedESS-TimeComplexity} and
    \cref{thm:SpaceBoundedUnbalancedESS}, for
    $1/4<\alpha\le\frac{2}{3}(\log 3-1)$.
    \item $\ell'_5(\alpha)$:
    \cref{thm:SpaceBoundedBalancedESS-TimeComplexity} and
    \cref{thm:SpaceBoundedUnbalancedESS}, for
    $\alpha_3<\alpha\le1/4$.
    \item $\ell'_6(\alpha)$:
    \cref{thm:SpaceBoundedBalancedESS-TimeComplexity} and
    \cref{thm:SpaceBounded-UseSubset-ESS}, for
    $\alpha_4<\alpha\le\alpha_3$.
    \item $\ell'_7(\alpha)$: \cref{thm:(0.5+l')-time} and
    \cref{thm:SpaceBounded-UseSubset-ESS}, for
    $0\le\alpha\le\alpha_4$.
\end{itemize}

Setting $\epsilon=1/100$, we have
$\ell'_i(\alpha)<1-\epsilon$ for all $i\in \{1,3,4,5,6\}$ and for every relevant
value of $\alpha$.
Therefore, we can safely apply
\cref{thm:RandomReplace-0.5-0.76-space,thm:SpaceBoundedBalancedESS-TimeComplexity}
as in \cref{tab:case-distinction}.
We also have $1/2<\ell'_2(\alpha)\le\gamma_1<0.8$, where $\gamma_1$ is
defined in \eqref{eq:gamma_1}.
Thus, we can apply \cref{thm:RandomReplace-0.45-0.5-space} as in
\cref{tab:case-distinction}.

Once we have determined which lemma or theorem to apply for each pair of
$\alpha$ and $\ell'$, we can compute, for each $\alpha$, the value of
$\ell'$ that maximizes the running time.
This is the bottleneck value of $\ell'$ shown in the last column of
\cref{tab:case-distinction}.
Using these values, we can express the worst-case running time under the
space bound $O^*(2^{\alpha n})$ as
$O^*(2^{\mathcal{T}(\alpha)n})$, where $\mathcal{T}(\alpha)$ is defined as
follows:
\begin{equation} \label{eq:mathcalT}
\mathcal{T}(\alpha)=
\begin{cases}
\frac{H(\gamma_1)+\gamma_1}{2} = \tau_{\mathrm{rrb}}(\gamma_1, \gamma_1 - \alpha_1) 
& \text{if } \alpha_1 < \alpha,\\[1mm]

\frac{H(\ell'_1(\alpha))+\ell'_1(\alpha)}{2}
=
\tau_{\mathrm{rrb}}(\ell'_1(\alpha),\ell'_1(\alpha)-\alpha)
& \text{if } 1/2<\alpha\le\alpha_1,\\[1mm]

\frac{H(\ell'_2(\alpha))+\ell'_2(\alpha)}{2}
=
\tau_{\mathrm{rrb}}\left(
    \ell'_2(\alpha),
    \ell'_2(\alpha)-\mathrm{phi}^{-1}(1-\alpha)
\right)
& \text{if } \alpha_2<\alpha\le1/2,\\[1mm]

\frac{H(\ell'_3(\alpha))+\ell'_3(\alpha)}{2}
=
\tau_{\mathrm{sbb}}(\ell'_3(\alpha),\alpha)
& \text{if } 2-\log 3<\alpha\le\alpha_2,\\[1mm]

\tau_{\mathrm{sbu}}\left(\frac{2}{3},\alpha\right)
& \text{if } \frac{2}{3}(\log 3-1)<\alpha\le 2-\log 3,\\[1mm]

\tau_{\mathrm{sbu}}(\ell'_4(\alpha),\alpha)
=
\tau_{\mathrm{sbb}}(\ell'_4(\alpha),\alpha)
& \text{if } 1/4<\alpha\le\frac{2}{3}(\log 3-1),\\[1mm]

\tau_{\mathrm{sbu}}(\ell'_5(\alpha),\alpha)
=
\tau_{\mathrm{sbb}}(\ell'_5(\alpha),\alpha)
& \text{if } \alpha_3<\alpha\le1/4,\\[1mm]

\tau_{\mathrm{sbb}}(\ell'_6(\alpha),\alpha)
=
H(\ell'_6(\alpha))+\tau_{\mathrm{subset}}(\alpha)\ell'_6(\alpha)
& \text{if } \alpha_4<\alpha\le\alpha_3,\\[1mm]

\ell'_7(\alpha)+\frac{1}{2}
=
H(\ell'_7(\alpha))+\tau_{\mathrm{subset}}(\alpha)\ell'_7(\alpha)
& \text{if } 0\le\alpha\le\alpha_4.
\end{cases}
\end{equation}
Here, $\gamma_1$, $\tau_{\mathrm{rrb}}$, $\tau_{\mathrm{sbu}}$,
$\tau_{\mathrm{sbb}}$, and $\tau_{\mathrm{subset}}$ are defined in
\eqref{eq:gamma_1}, \eqref{eq:tau-rrb}, \eqref{eq:tau_sbu},
\eqref{eq:tau_sbb}, and \eqref{eq:tau-subset}, respectively.
The function $\mathrm{phi}^{-1}$ is defined in
\cref{thm:RandomReplace-0.45-0.5-space}.
For the case $\alpha>\alpha_1$, we use
\cref{thm:MainResult-ExpSpace}, whose space complexity is bounded by
$O^*(2^{\alpha_1 n})\le O^*(2^{\alpha n})$.
This yields \cref{thm:MainResult-TimeSpaceTradeoff}, restated below.
\TimeSpaceTradeoffThm*

\subsection{Space-Bounded Algorithms via Theorem~\ref{thm:FastExpSpaceComplexity}} \label{sec:5-2}

We use \cref{thm:FastExpSpaceComplexity} to obtain algorithms with
smaller space bounds.

\begin{theorem} \label{thm:RandomReplace-0.5-0.76-space}
Assume that the minimum solution ratio satisfies
$\ell'\in(1/2,1-\epsilon]$ for some constant $\epsilon>0$.
For every constant $\alpha$ satisfying $1/2<\alpha\le\alpha_1$, where
$\alpha_1$ is defined in \eqref{eq:alpha_1}, there is a Monte Carlo
algorithm for ESS that uses $O^*(2^{\alpha n})$ space and runs in
$O^*(2^{\tau_{\mathrm{rrb}}(\ell',\max\{0,\ell'-\alpha\})n})$ time.
\end{theorem}

\begin{proof}
Set $k'=\max\{0,\ell'-\alpha\}$.
Then $k'\ge0$.
Moreover, since $\alpha>1/2$ and $\ell'<1$, we have $2k' \le \ell'$, and we can apply
\cref{thm:FastExpSpaceComplexity}.
Since $\ell'-k'=\min\{\ell',\alpha\} > 1/2$, the space exponent is
\[
\sigma_{\mathrm{rrb}}(\ell',k')
=
\ell'-k'
=
\min\{\ell',\alpha\}
\le
\alpha.
\]
Therefore, the space complexity is
$O^*(2^{\sigma_{\mathrm{rrb}}(\ell',k')n})\le O^*(2^{\alpha n})$.
The claimed running time follows directly from
\cref{thm:FastExpSpaceComplexity} with the above choice of $k'$.
\end{proof}

\begin{theorem} \label{thm:RandomReplace-0.45-0.5-space}
Assume that the minimum solution ratio satisfies
$\ell'\in(1/2,0.8)$.
For every constant $\alpha$ satisfying $0.45\le \alpha\le 1/2$, there is a
Monte Carlo algorithm for ESS that uses $O^*(2^{\alpha n})$ space and runs in
$O^*(2^{\tau_{\mathrm{rrb}}(\ell',
\ell'-\mathrm{phi}^{-1}(1-\alpha))n})$ time.
Here, $\mathrm{phi}^{-1}$ denotes the inverse of the restriction of
$\mathrm{phi}$ defined in \eqref{eq:phi(x)} to the interval $[0.4,1/2]$,
on which $\mathrm{phi}$ is bijective.
\end{theorem}
\begin{proof}
Set $k'=\ell'-\mathrm{phi}^{-1}(1-\alpha)$.
Since $\alpha\in[0.45,1/2]$, we have
$\mathrm{phi}^{-1}(1-\alpha)\in[0.4,1/2]$.
Also, $k'\ge0$ follows from $\ell'>1/2$ and
$\mathrm{phi}^{-1}(1-\alpha)\le1/2$.
Moreover, since $\ell'<0.8$ and
$\mathrm{phi}^{-1}(1-\alpha)\ge0.4$, we have
$2k'=2(\ell'-\mathrm{phi}^{-1}(1-\alpha))\le\ell'$.
Thus, the condition $0\le k'\le \ell'/2$ is satisfied, and we can apply
\cref{thm:FastExpSpaceComplexity}.
Since $\ell'-k'=\mathrm{phi}^{-1}(1-\alpha)\le1/2$, the space exponent is
\[
    \sigma_{\mathrm{rrb}}(\ell',k')
    =
    1-\mathrm{phi}(\ell'-k')
    =
    1-\mathrm{phi}(\mathrm{phi}^{-1}(1-\alpha))
    =
    \alpha.
\]
Therefore, the space complexity is
$O^*(2^{\sigma_{\mathrm{rrb}}(\ell',k')n})=O^*(2^{\alpha n})$.
The claimed running time follows directly from
\cref{thm:FastExpSpaceComplexity} with the above choice of $k'$.
\end{proof}

\subsection{Space-Bounded Modification of the Balanced ESS Algorithm} \label{sec:5-3}
\begin{algorithm}[htbp]
\caption{SpaceBoundedBalancedESS}
\DontPrintSemicolon
\SetKwInput{KwSpace}{Space bound}
\KwIn{Input set $S$ and minimum solution size $\ell$}
\KwOut{An ESS solution for $S$ if one exists; otherwise, NO}
\KwSpace{$O^*(2^{\alpha n})$}

Set
\[
p_{\max}
=
\begin{cases}
\min\{2^{n-\ell},\left\lfloor2^{\alpha n}\right\rfloor\}
& \text{if } \ell>n/2,\\[1mm]
\min\left\{
    \binom{n-\ell}{\left\lfloor(\ell-1)/2\right\rfloor},
    \left\lfloor2^{\alpha n}\right\rfloor
\right\}
& \text{otherwise.}
\end{cases}
\]
Pick a random prime $p$ in $[p_{\max},2p_{\max}]$, and a random residue
$r$ in $[0,p-1]$.\;

Let $C=\{X\subseteq S\mid \sum(X)\equiv_p r\}$.
Prepare the data structure $DS$ of \cref{lm:Fast-Subset-Sum-Oracle} in
$O(np)$ time and space, so that $DS$ supports random access to the
elements of $C$ and stores the value $|C|$.
If $|C|>n^d\cdot 2^n/p_{\max}$ for a fixed constant $d>2$, halt and
return \textnormal{NO}.\;

Consider the list $L$ whose $i$-th element is the sum of the $i$-th
element of $C$.
Using $DS$, each entry of $L$ can be computed on the fly in
$\mathrm{poly}(n)$ time.
Apply \cref{lm:ElementDistinctNess-TimeSpaceTradeoff} to $L$ to find
$X,Y\in C$ such that $\sum(X)=\sum(Y)$, using
$O^*(\min\{|C|,2^{\alpha n}\})$ space.
If no such pair exists, return \textnormal{NO}.\;
Return $(X\setminus Y,\,Y\setminus X)$.\;
\label{alg:SpaceBoundedBalancedESS}
\end{algorithm}
We introduce \cref{alg:SpaceBoundedBalancedESS}, which combines the
\emph{BalancedEqualSubsetSum} algorithm of
\cite{mucha_et_al:LIPIcs.ESA.2019.73} with a low-space algorithm for
Element Distinctness.
We use the following time-space tradeoff for Element Distinctness due to
\cite{lyu2023time}.
Although the statement in \cite{lyu2023time} is phrased for positive
integers, the same bound applies to nonnegative integers by adding $1$ to every array entry, which preserves all collisions.

\begin{lemma}[\cite{lyu2023time}] \label{lm:ElementDistinctNess-TimeSpaceTradeoff}
Given random access to an array of $N$ nonnegative integers
$a_1,\ldots,a_N$, with $a_i \le \mathrm{poly}(N)$ for all $i\in[N]$,
there is a Monte Carlo algorithm that either returns a colliding pair
$(i,j)$ such that $i\neq j$ and $a_i=a_j$, or reports that no such pair
exists.
The algorithm runs in $O(t(N)\mathrm{polylog}(N))$ time and uses
$O(s(N)\mathrm{polylog}(N))$ space, for any complexity bounds
$s,t\colon \mathbb{N}\to\mathbb{N}$ satisfying
$s(N)\le N$ and $s(N)^{1/2}\cdot t(N)\ge N^{3/2}$.
\end{lemma}

We use the \emph{Fast Subset-Sum Oracle} of
\cite[Section~3.1]{allcock_et_al:LIPIcs.ESA.2022.6} to support random access to the
elements of the set $C$ in \cref{alg:SpaceBoundedBalancedESS}.

\begin{lemma}[{\cite[Section~3.1]{allcock_et_al:LIPIcs.ESA.2022.6}}]
\label{lm:Fast-Subset-Sum-Oracle}
Let $\prec$ be the relation on $\{I\mid I\subseteq[n]\}$ defined as
follows: for all $I_1,I_2\subseteq[n]$, we have $I_1\prec I_2$ if and
only if $\max\{i\mid i\in I_1\Delta I_2\}\in I_2$.
Then $\prec$ is a strict total order.
Let $S=\{w_1,\ldots,w_n\}$ be an input set, and let
$C_{\mathrm{idx}}=\{I\subseteq[n]\mid \sum_{i\in I} w_i\equiv_p r\}$.
In $O(np)$ time and space, we can construct a data structure that, given
an index $J\in[|C_{\mathrm{idx}}|]$, returns the $J$-th subset in
$C_{\mathrm{idx}}$ with respect to $\prec$ in $O(n)$ time.
The data structure also stores the value $|C_{\mathrm{idx}}|$.
\end{lemma}

Since the input $S$ is a set, we identify each index set $I\subseteq[n]$
with the corresponding subset $\{w_i\mid i\in I\}$ of $S$.
Under this identification, \cref{lm:Fast-Subset-Sum-Oracle} provides
random access to
$C=\{X\subseteq S\mid \sum(X)\equiv_p r\}$.
The success-probability analysis of
\cref{alg:SpaceBoundedBalancedESS} is essentially the same as that of
\cref{alg:FastExpSpace} with replacement count $k=0$.

\begin{lemma} \label{lm:SpaceBoundedBalancedESS-SuccessProb}
Assume that $\epsilon n\le \ell\le (1-\epsilon)n$ for some constant
$\epsilon>0$, and let $\alpha\in(0,1/2]$ be a constant.
Then \cref{alg:SpaceBoundedBalancedESS} returns a solution with
probability at least $\Omega(1/n^2)$.
\end{lemma}

\begin{proof}
We follow the same argument as in the proof of
\cref{lm:FastAlgSuccessProb}, except that we do not modify the input set
using \cref{alg:RandomReplace}.
We first state the special cases of the claims used there corresponding to $k=0$, where no input modification is performed.

\begin{claim} \label{clm:Sec5-Phi-Size}
Let
\[
\Phi
=
\left\{
    \sum(X)
    \mid
    X\subseteq S,\ \exists Y\subseteq S
    \text{ such that } X\neq Y \text{ and } \sum(X)=\sum(Y)
\right\}.
\]
Then
\[
|\Phi| \ge
\begin{cases}
2^{n-\ell}
& \text{if } \ell>n/2,\\[1mm]
\binom{n-\ell}{\left\lfloor(\ell-1)/2\right\rfloor}
& \text{otherwise.}
\end{cases}
\]
\end{claim}
\begin{proof}
The claim follows by applying the proof of \cref{lm:Size-of-Phi} with
$k=0$.
\end{proof}

\begin{claim} \label{clm:Sec5-Phi-mod-p-Size}
Suppose that $|\Phi|\ge p_{\max}$ and $p_{\max}=2^{\Theta(n)}$.
Then, with constant probability, we have
$|\{a\bmod p\mid a\in\Phi\}|\ge \Omega(p/n^2)$.
\end{claim}
\begin{proof}
The claim follows from the proof of \cref{clm:Lower-Bound-Phi-mod-p}.
\end{proof}

\begin{claim} \label{clm:Sec5-Exp-Size-C}
We have $\mathbb{E}[|C|]\le 2^n/p_{\max}$.
\end{claim}
\begin{proof}
The claim follows from the proof of \cref{clm:Exp-Size-Of-C}.
\end{proof}

Let $\mathcal{G}$ denote the event that
\cref{alg:SpaceBoundedBalancedESS} does not halt in Line~2, i.e., the
event that $|C|\le n^d \frac{2^n}{p_{\max}}$ occurs, where $d>2$ is the
fixed constant used in Line~2.
When $\mathcal{G}$ occurs, the algorithm reaches Line~3.
Let $\mathcal{H}$ denote the event that there exist $X,Y\in C$ satisfying
$\sum(X)=\sum(Y)$.
When $\mathcal{H}$ occurs,
\cref{lm:ElementDistinctNess-TimeSpaceTradeoff} finds such a pair in
Line~3 with constant probability.
Therefore, the probability that
\cref{alg:SpaceBoundedBalancedESS} outputs a solution is at least a
constant factor times $\mathbb{P}[\mathcal{G}\cap\mathcal{H}]$.

By Line~1 of \cref{alg:SpaceBoundedBalancedESS} and
\cref{clm:Sec5-Phi-Size}, we have $p_{\max}\le |\Phi|$.
From the assumptions on $\ell$ and $\alpha$, we also have
$p_{\max}=2^{\Theta(n)}$.
Therefore, by \cref{clm:Sec5-Phi-mod-p-Size}, with constant probability,
$|\{a\bmod p\mid a\in\Phi\}|\ge \Omega(p/n^2)$.
Conditioned on this event, the residue $r$ chosen in Line~1 of
\cref{alg:SpaceBoundedBalancedESS} belongs to
$\{a\bmod p\mid a\in\Phi\}$ with probability at least $\Omega(1/n^2)$.
Thus, we have
\begin{equation} \label{eq:Sec5-P[H]}
\mathbb{P}[\mathcal{H}]\ge \Omega(1/n^2).
\end{equation}

Let $\mathcal{G}^c$ denote the complement of $\mathcal{G}$, namely the
event that $|C|>n^d\frac{2^n}{p_{\max}}$ occurs.
By \cref{clm:Sec5-Exp-Size-C} and Markov's inequality, we have
\begin{equation} \label{eq:Sec5-P[G^c]}
\mathbb{P}[\mathcal{G}^c]
\le
\frac{\mathbb{E}[|C|]}{n^d\cdot 2^n/p_{\max}}
\le
\frac{1}{n^d}.
\end{equation}
Combining \eqref{eq:Sec5-P[H]} and \eqref{eq:Sec5-P[G^c]}, we obtain
\[ \mathbb{P}[\mathcal{G}\cap\mathcal{H}] = \mathbb{P}[\mathcal{H}] - \mathbb{P}[\mathcal{G}^c\cap\mathcal{H}] \ge \Omega(1/n^2)-1/n^d = \Omega(1/n^2), \]
since $d>2$.
This completes the proof of
\cref{lm:SpaceBoundedBalancedESS-SuccessProb}.
\end{proof}

We now turn to the time and space complexity analysis of
\cref{alg:SpaceBoundedBalancedESS}.

\begin{theorem} \label{thm:SpaceBoundedBalancedESS-TimeComplexity}
Assume that the minimum solution ratio satisfies
$\ell'\in[\epsilon,1-\epsilon]$ for some constant $\epsilon>0$.
For every constant $\alpha \in (0, 1/2]$, there
is a Monte Carlo algorithm for ESS that uses $O^*(2^{\alpha n})$ space
and runs in $O^*(2^{\tau_{\mathrm{sbb}}(\ell',\alpha)n})$ time.
Here, $\mathrm{phi}$ is defined in \eqref{eq:phi(x)}, and
$\tau_{\mathrm{sbb}}(\ell',\alpha)$ is defined as follows:
\begin{equation} \label{eq:tau_sbb}
    \tau_{\mathrm{sbb}}(\ell',\alpha)
    =
    \begin{cases}
        \frac{3}{2}(1-\min\{\alpha,1-\ell'\})-\frac{\alpha}{2}
        & \text{if } \ell'>\frac{1}{2}, \\[1mm]
        \frac{3}{2}(1-\min\{\alpha,\mathrm{phi}(\ell')\})-\frac{\alpha}{2}
        & \text{otherwise.}
    \end{cases}
\end{equation}
\end{theorem}
\begin{proof}
We show that \cref{alg:SpaceBoundedBalancedESS} satisfies these
complexity bounds.
Line~2 of \cref{alg:SpaceBoundedBalancedESS} takes
$O(np)=O^*(p_{\max})$ time and space.
Since $p_{\max}\le 2^{\alpha n}$ by the definition of $p_{\max}$, the
space usage of Line~2 is bounded by $O^*(2^{\alpha n})$.
Moreover, the space usage of Line~3 is bounded by
$O^*(\min\{|C|,2^{\alpha n}\})\le O^*(2^{\alpha n})$.
Therefore, the overall space complexity is bounded by
$O^*(2^{\alpha n})$.

Let $L$ be the list considered in Line~3.
Since $|L|=|C|\le O^*(2^n/p_{\max})$, we have
$\mathrm{polylog}(|L|)=\mathrm{poly}(n)$.
Thus, the $\mathrm{polylog}(|L|)$ factor can be absorbed into the
$O^*(\cdot)$ notation.

By \cref{lm:ElementDistinctNess-TimeSpaceTradeoff}, Line~3 takes
\[
O^*\left(
    \frac{|C|^{3/2}}
    {(\min\{|C|,2^{\alpha n}\})^{1/2}}
\right)
\le
O^*\left(
    \frac{(2^n/p_{\max})^{3/2}}{2^{\alpha n/2}}
\right)
\]
time.
Here, the inequality follows from $|C|\le O^*(2^n/p_{\max})$ and
$2^{\alpha n}\le 2^n/p_{\max}$, where the latter holds because
$p_{\max}\le 2^{\alpha n}$ and $\alpha\le 1/2$.
Moreover, since $p_{\max}\le 2^n/p_{\max}\le
\frac{(2^n/p_{\max})^{3/2}}{2^{\alpha n/2}}$, the running time of Line~2 is
dominated by that of Line~3.
By \cref{lm:SpaceBoundedBalancedESS-SuccessProb}, repeating the algorithm
$\mathrm{poly}(n)$ times amplifies the success probability to a constant,
and this repetition does not affect the time complexity in
$O^*(\cdot)$ notation.
Thus, the overall time complexity is
$O^*\left(\frac{(2^n/p_{\max})^{3/2}}{2^{\alpha n/2}}\right)$.

By the definition of $p_{\max}$, we have
\[
p_{\max}
=
\begin{cases}
\Theta^*(2^{\min\{\alpha,1-\ell'\}n})
& \text{if } \ell'>1/2, \\[1mm]
\Theta^*(2^{\min\{\alpha,\mathrm{phi}(\ell')\}n})
& \text{otherwise.}
\end{cases}
\]
Here, we used
$\binom{n-\ell}{\left\lfloor(\ell-1)/2\right\rfloor}
=
\Theta^*(2^{\mathrm{phi}(\ell')n})$,
which follows from
$\left|
\frac{\left\lfloor(\ell-1)/2\right\rfloor}{n-\ell}
-
\frac{\ell'}{2(1-\ell')}
\right|
\le O(1/n)$
and \cref{fact:H(x+1/n)-H(x)}, using the assumption
$\ell'\in[\epsilon,1-\epsilon]$, which implies $\ell=\Theta(n)$ and
$n-\ell=\Theta(n)$.
Therefore, we obtain
$O^*\left(\frac{(2^n/p_{\max})^{3/2}}{2^{\alpha n/2}}\right) = O^*(2^{\tau_{\mathrm{sbb}}(\ell',\alpha)n})$, which completes the proof.
\end{proof}

\subsection{Space Bounded Modificiation of the Unbalanced ESS Algorithm} \label{sec:5-4}

\begin{algorithm}[htbp]
\caption{SpaceBoundedUnbalancedESS}
\DontPrintSemicolon
\SetKwInput{KwSpace}{Space bound}
\KwIn{Input set $S$ and minimum solution size $\ell$}
\KwOut{An ESS solution for $S$ if one exists; otherwise, \textnormal{NO}}
\KwSpace{$O^*(2^{\alpha n})$}

Let $\ell'=\ell/n$.
If $\frac{1}{4}(H(\ell')+\ell')\le \alpha$, then apply
\cref{lm:UnbalancedESS} to the whole input set $S$ and return its output.
Otherwise, set $\beta=\frac{4\alpha}{H(\ell')+\ell'} < 1$.
Randomly split the input set $S$ into two disjoint sets
$S_1, S_2 \subseteq S$ such that
$|S_1|=\lfloor\beta n\rfloor$, and
$|S_2|=n-\lfloor\beta n\rfloor$.\;

Generate the triples $(c,X_2,Y_2)$ one by one, where
$X_2,Y_2\subseteq S_2$, $X_2\cap Y_2=\emptyset$,
$|X_2|+|Y_2|=\ell-\lfloor\beta\ell\rfloor$, and
$c=\sum(X_2)-\sum(Y_2)$.
We do not store all such triples simultaneously; instead, each generated
triple is processed immediately.
For each generated triple $(c,X_2,Y_2)$, apply \cref{lm:TargetESS} to
the input set $S_1$ with target integer $\kappa=-c$ and minimum solution
size parameter $\lfloor\beta\ell\rfloor$.
If \cref{lm:TargetESS} finds subsets $X_1,Y_1\subseteq S_1$ such that
$\sum(X_1)-\sum(Y_1)=-c$, then return
$(X_1\cup X_2,\,Y_1\cup Y_2)$.\;

Return \textnormal{NO}.\;

\label{alg:SpaceBoundedUnbalancedESS}
\end{algorithm}

We introduce \cref{alg:SpaceBoundedUnbalancedESS}, a space-bounded
variant of the algorithm used in \cref{lm:UnbalancedESS}.
This algorithm uses \cref{lm:TargetESS} as a subroutine.
The \emph{Target Equal-Subset-Sum Problem} is a slightly generalized
version of ESS.
As stated in \cref{lm:TargetESS}, it can be solved within the same time
and space bounds as those in \cref{lm:UnbalancedESS}
\cite[Appendix~E of the full version]{mucha_et_al:LIPIcs.ESA.2019.73}.
We provide the details in Appendix~\ref{app:Target-ESS}.

\begin{restatable}[{\cite[Appendix~E of the full version]{mucha_et_al:LIPIcs.ESA.2019.73}}]{lemma}{targetess}
\label{lm:TargetESS}
The \emph{Target Equal-Subset-Sum (Target ESS) Problem} is defined as
follows.
Given an input set $S$ of $n$ integers and a target integer $\kappa$, the
task is to find two distinct subsets $A,B\subseteq S$ such that
$\sum(A)-\sum(B)=\kappa$, or to report that no such pair of subsets
exists.
If such subset pairs exist, let $\ell$ denote the minimum value of
$|A|+|B|$ over all pairs of distinct subsets $A,B\subseteq S$ satisfying
$\sum(A)-\sum(B)=\kappa$.
We define a minimum solution of the Target ESS instance to be a pair
$(A,B)$ attaining this minimum, and let $\ell'=\ell/n$ denote the minimum
solution ratio of the Target ESS instance.
Then, for any Target ESS instance $(S,\kappa)$ with minimum solution ratio
$\ell'$, there is a Monte Carlo algorithm that solves the
\emph{Target ESS Problem} in
$O^*(2^{\frac{n}{2}(H(\ell')+\ell')})$ time and
$O^*(2^{\frac{n}{4}(H(\ell')+\ell')})$ space.
\end{restatable}

\cref{alg:SpaceBoundedUnbalancedESS} relies on the assumption that
$S_1$ contains roughly a $\beta$-fraction of the elements of a minimum
solution, just as it contains a $\beta$-fraction of the input set.
We show that such a split occurs with probability at least
$\Omega(1/\mathrm{poly}(n))$.

\begin{lemma} \label{lm:Split-Spacebounded-UnbalancedESS}
Let $A,B\subseteq S$ be a minimum solution with $|A\cup B|=\ell$.
Then
$|(A\cup B)\cap S_1|=\lfloor\beta\ell\rfloor$ and
$|(A\cup B)\cap S_2|=\ell-\lfloor\beta\ell\rfloor$
hold with probability at least $\Omega(1/\mathrm{poly}(n))$.
\end{lemma}

\begin{proof}
The probability that such a good split occurs is
$\frac{
    \binom{\ell}{\lfloor\beta\ell\rfloor}
    \binom{n-\ell}{\lfloor\beta n\rfloor-\lfloor\beta\ell\rfloor}
}{
    \binom{n}{\lfloor\beta n\rfloor}
}$.
The denominator satisfies
$\binom{n}{\lfloor\beta n\rfloor}=\Theta^*(2^{H(\beta)n})$, since
$\left|\frac{\lfloor\beta n\rfloor}{n}-\beta\right| \le O(1/n)$, together
with \cref{fact:H(x+1/n)-H(x)}.

We show that the numerator is also $\Theta^*(2^{H(\beta)n})$.
If $\ell=O(1)$, then $\binom{\ell}{\lfloor\beta\ell\rfloor}$ is
polynomially bounded, and
$
    \binom{n-\ell}{\lfloor\beta n\rfloor-\lfloor\beta\ell\rfloor}
    =
    \Theta^*(2^{H(\beta)n}),
$
since
$
\left|
\frac{\lfloor\beta n\rfloor-\lfloor\beta\ell\rfloor}{n-\ell}
-\beta
\right|
\le
O(1/n),
$
together with \cref{fact:H(x+1/n)-H(x)}.
Similarly, if $\ell=n-O(1)$, then
$\binom{n-\ell}{\lfloor\beta n\rfloor-\lfloor\beta\ell\rfloor}$ is
polynomially bounded, and
$
    \binom{\ell}{\lfloor\beta\ell\rfloor}
    =
    \Theta^*(2^{H(\beta)n}),
$
since
$\left|\frac{\lfloor\beta\ell\rfloor}{\ell}-\beta\right| \le O(1/n)$,
together with \cref{fact:H(x+1/n)-H(x)}.
Otherwise, we have
\begin{align}
\binom{\ell}{\lfloor\beta\ell\rfloor}
    \binom{n-\ell}{\lfloor\beta n\rfloor-\lfloor\beta\ell\rfloor}
&=
\Theta^*(2^{H(\frac{\lfloor\beta\ell\rfloor}{\ell})\ell})
    \Theta^*(2^{H(\frac{\lfloor\beta n\rfloor-\lfloor\beta\ell\rfloor}{n-\ell})(n-\ell)}) \\
&=
\Theta^*(2^{H(\beta)\ell})
    \Theta^*(2^{H(\beta)(n-\ell)})
=
\Theta^*(2^{H(\beta)n}).
\end{align}
The second equality follows from
$\left|\frac{\lfloor\beta\ell\rfloor}{\ell}-\beta\right| \le O(1/\ell)$
and
$\left|\frac{\lfloor\beta n\rfloor-\lfloor\beta\ell\rfloor}{n-\ell}
-\beta\right| \le O(1/(n-\ell))$, together with
\cref{fact:H(x+1/n)-H(x)}.

Thus, in every case, the numerator is
$\Theta^*(2^{H(\beta)n})$.
Therefore, the good split occurs with probability $\Theta^*(1)$, and in
particular with probability at least $1/\mathrm{poly}(n)$.
This completes the proof.
\end{proof}

We now turn to the time and space complexity analysis of
\cref{alg:SpaceBoundedUnbalancedESS}.

\begin{theorem} \label{thm:SpaceBoundedUnbalancedESS}
Let $\ell'$ be the minimum solution ratio.
For every constant $\alpha>0$, there is a Monte Carlo algorithm for ESS
that uses $O^*(2^{\alpha n})$ space and runs in
$O^*(2^{\tau_{\mathrm{sbu}}(\ell',\alpha)n})$ time, where
$\tau_{\mathrm{sbu}}(\ell',\alpha)$ is defined by
\begin{equation} \label{eq:tau_sbu}
\tau_{\mathrm{sbu}}(\ell',\alpha)
=
\left(
1-\frac{\min\left\{\frac{4\alpha}{H(\ell')+\ell'},1\right\}}{2}
\right)
(H(\ell')+\ell').
\end{equation}
\end{theorem}

\begin{proof}
If $(H(\ell')+\ell')/4\le \alpha$, then by
\cref{lm:UnbalancedESS}, the algorithm uses
$O^*(2^{\frac{n}{4}(H(\ell')+\ell')})\le O^*(2^{\alpha n})$ space and
runs in $O^*(2^{\frac{n}{2}(H(\ell')+\ell')})$ time.
This is consistent with the definition of
$\tau_{\mathrm{sbu}}(\ell',\alpha)$, since
$\min\left\{\frac{4\alpha}{H(\ell')+\ell'},1\right\}=1$ in this case.
Thus, it remains to consider the case $(H(\ell')+\ell')/4>\alpha$.

By \cref{lm:Split-Spacebounded-UnbalancedESS}, repeating the
algorithm $\mathrm{poly}(n)$ times amplifies to a constant the probability
that
\[
    |(A\cup B)\cap S_1|=\lfloor\beta\ell\rfloor
    \quad\text{and}\quad
    |(A\cup B)\cap S_2|=\ell-\lfloor\beta\ell\rfloor,
\]
where $A,B\subseteq S$ are a minimum solution.
When this good split occurs, Line~2 of
\cref{alg:SpaceBoundedUnbalancedESS} eventually considers
$X_2=A\cap S_2$ and $Y_2=B\cap S_2$, since it enumerates all subsets
$X_2,Y_2\subseteq S_2$ such that $X_2\cap Y_2=\emptyset$ and
$|X_2|+|Y_2|=\ell-\lfloor\beta\ell\rfloor$.
This enumeration takes
\[
O\left(
    \binom{n-\lfloor\beta n\rfloor}
          {\ell-\lfloor\beta\ell\rfloor}
    2^{\ell-\lfloor\beta\ell\rfloor}
\right)
=
O^*(2^{(H(\ell')+\ell')(1-\beta)n})
\]
time and polynomial space, since the algorithm stores only one candidate pair $(X_2,Y_2)$ at a time, and $\binom{n-\lfloor\beta n\rfloor}{\ell-\lfloor\beta\ell\rfloor} =O^*(2^{H(\ell')(1-\beta)n})$.
To see this binomial-coefficient bound, if $n-\lfloor\beta n\rfloor=O(1)$, then the binomial coefficient is polynomially bounded. Otherwise, we have $\left| \frac{\ell-\lfloor\beta\ell\rfloor}{n-\lfloor\beta n\rfloor} - \ell' \right| \le O\left(\frac{1}{(1-\beta)n}\right)$, and the claim follows from \cref{fact:H(x+1/n)-H(x)}.

For each generated pair $X_2,Y_2$, we apply \cref{lm:TargetESS} to
$S_1$ with target integer $\kappa=\sum(Y_2)-\sum(X_2)$.
When the good split occurs and the enumeration reaches
$X_2=A\cap S_2$ and $Y_2=B\cap S_2$, the subsets
$A\cap S_1$ and $B\cap S_1$ form a minimum solution to this Target ESS
instance, with minimum solution size
$|A\cap S_1|+|B\cap S_1|=\lfloor\beta\ell\rfloor$.
Thus, applying \cref{lm:TargetESS} with this minimum solution size
parameter, we find, with constant probability, subsets
$X_1,Y_1\subseteq S_1$ such that
$\sum(X_1)-\sum(Y_1)=\sum(Y_2)-\sum(X_2)$ in
\[
O^*\left(
2^{
\frac{\lfloor\beta n\rfloor}{2}
\left(
H\left(\frac{\lfloor\beta\ell\rfloor}{\lfloor\beta n\rfloor}\right)
+
\frac{\lfloor\beta\ell\rfloor}{\lfloor\beta n\rfloor}
\right)
}
\right)
=
O^*\left(2^{\frac{\beta}{2}(H(\ell')+\ell')n}\right)
\]
time. The equality follows from $\left|\frac{\lfloor\beta\ell\rfloor}{\lfloor\beta n\rfloor} -\ell'\right|\le O(1/n)$ and \cref{fact:H(x+1/n)-H(x)}, since $\beta$ is lower-bounded by a constant $\alpha>0$.
Then $\sum(X_1)+\sum(X_2)=\sum(Y_1)+\sum(Y_2)$.
Since $S_1$ and $S_2$ are disjoint, we also have
$X_1\cap X_2=\emptyset$ and $Y_1\cap Y_2=\emptyset$.
Thus, $(X_1\cup X_2,\,Y_1\cup Y_2)$ is a valid ESS solution.

The overall time complexity is
\[
O^*(2^{(H(\ell')+\ell')(1-\beta)n})
\cdot
O^*\left(2^{\frac{\beta}{2}(H(\ell')+\ell')n}\right)
=
O^*(2^{(1-\frac{\beta}{2})(H(\ell')+\ell')n}).
\]
In this case, we have
$\beta=\frac{4\alpha}{H(\ell')+\ell'}<1$.
Thus, the above bound is consistent with the definition of
$\tau_{\mathrm{sbu}}(\ell',\alpha)$.
This completes the proof.
\end{proof}

\subsection{Space-Bounded Algorithm via Subset-Sum} \label{sec:5-5}
We obtain another space-bounded algorithm for ESS by replacing the
polynomial-space Subset-Sum algorithm used in
\cref{thm:Use-SubsetSum-As-Routine} with a space-bounded Subset-Sum
algorithm.
Specifically, we use the faster of the time-space tradeoff algorithm for
Subset Sum due to \cite{austrin2013space} and the polynomial-space
algorithm of \cite{Bansal2018-FastPolySpaceSubsetSum,chen2022truly}.

\begin{lemma}[\cite{austrin2013space,Bansal2018-FastPolySpaceSubsetSum,chen2022truly}]
\label{lm:Subset-Sum-Time-Space-Tradeoff}
Define $\rho(i)=1+i(i+1)/2$ for integers $i\ge1$.
For every constant $\alpha\in(0,1/2]$, let $i_\alpha$ be the integer such that $\frac{1}{\rho(i_\alpha+1)}<\alpha\le\frac{1}{\rho(i_\alpha)}$,
and define
\begin{equation}
\label{eq:tau-subset}
\tau_{\mathrm{subset}}(\alpha)
=
\min\left\{
1-\frac{1}{i_\alpha+1}
-\frac{\rho(i_\alpha)-2}{i_\alpha+1}\alpha,
0.86
\right\}.
\end{equation}
We also define $\tau_{\mathrm{subset}}(0)=0.86$.
Then there is a Monte Carlo algorithm for Subset-Sum that runs in
$O^*(2^{\tau_{\mathrm{subset}}(\alpha)n})$ time and uses
$O^*(2^{\alpha n})$ space.
\end{lemma}

We use \cref{lm:Subset-Sum-Time-Space-Tradeoff} to obtain the following
theorem.

\begin{theorem} \label{thm:SpaceBounded-UseSubset-ESS}
Let $\ell'$ be the minimum solution ratio.
For every constant $\alpha \in [0, 1/2]$, there is a Monte Carlo algorithm for ESS
that runs in $O^*(2^{(H(\ell')+\tau_{\mathrm{subset}}(\alpha)\ell')n})$
time and uses $O^*(2^{\alpha n})$ space.
\end{theorem}

\begin{proof}
We replace the $O^*(2^{0.86n})$-time polynomial-space Subset-Sum
algorithm used in the proof of \cref{thm:Use-SubsetSum-As-Routine} with
the Subset-Sum algorithm from
\cref{lm:Subset-Sum-Time-Space-Tradeoff}, which runs in
$O^*(2^{\tau_{\mathrm{subset}}(\alpha)n})$ time and uses
$O^*(2^{\alpha n})$ space.
Since the rest of the algorithm and its analysis remain unchanged, this
replacement gives the claimed time and space bounds.
\end{proof}

\section*{Acknowledgements}
The authors used Mathematica for numerical evaluations.
The authors also used ChatGPT for grammar correction and language
polishing.
All AI-assisted text was reviewed by the authors, who assume full
responsibility for the content of the paper.

\printbibliography

\appendix

\section{Proof of Lemma~\ref{lm:preprocess-inputs}}
\label{app:A}

We use the following lemma, which gives a lower bound on the number of
primes in an interval.
Recall that throughout the paper, all logarithms are base $2$.

\begin{lemma}[{\cite[p.371]{hardy1979introduction}}]
\label{lm:prime-density}
For all sufficiently large $x$, the interval $[x,2x]$ contains
$\Omega(x/\log x)$ primes.
\end{lemma}

We restate \cref{lm:preprocess-inputs}.
\PreprocessInputs*

We essentially repeat the argument from
\cite[Appendix~A of the full version]{mucha_et_al:LIPIcs.ESA.2019.73}.
If $0\in S$, then we can immediately return an ESS solution
$A=\{0\}$ and $B=\emptyset$.
If $m\ge 2^n$, then an algorithm running in time
$O(mn4^n)\le O(nm^3)$ is polynomial in $n$ and $m$.
Thus, we may assume that $0\notin S$ and $m<2^n$.

We write $S=\{w_1,\ldots,w_n\}$ and construct a new set
$S'=\{w'_1,\ldots,w'_{n+\lceil\log n\rceil}\}$ as follows.
We pick a random prime $p\in[2^{bn},2^{bn+1}]$ for a constant
$b\ge 6$.
For each $i\in[n]$, we set $w'_i = w_i \bmod p$.
We then add $\lceil\log n\rceil$ auxiliary elements by setting
$w'_{n+j}=2^{j-1}p$ for each
$j\in[\lceil\log n\rceil]$.

Since $p=2^{O(n)}$, we have $0\le s'\le 2^{O(n)}$ for all $s'\in S'$.
The bad event that either some constructed value is $0$ or two constructed
values coincide is included in the event that $S'$ has an ESS solution
that does not correspond to any ESS solution of $S$.
We will later show that the probability of this latter event is
$O(2^{-n})$.

For every $X,Y\subseteq[n]$, if
$\sum_{i\in X} w_i=\sum_{i\in Y} w_i$, then
$\sum_{i\in X} w'_i\equiv \sum_{i\in Y} w'_i \pmod p$.
Hence, $\left|\sum_{i\in X} w'_i-\sum_{i\in Y} w'_i\right| < np$
is a multiple of $p$.
Therefore, there exists an integer $r$ with $0\le r < n$ such that
\[
    \left|\sum_{i\in X} w'_i-\sum_{i\in Y} w'_i\right|=rp.
\]
Since the auxiliary elements are
$p,2p,4p,\dots,2^{\lceil\log n\rceil-1}p$, every multiple $rp$ with
$0\le r < n$ can be represented as the sum of a subset of the auxiliary
elements, by the binary representation of $r$.
Thus, if $S$ has an ESS solution, then $S'$ also has an ESS solution.

It remains to bound the probability that $S'$ contains an ESS solution
that does not correspond to any ESS solution of $S$.
First, there is no ESS solution of $S'$ that uses only auxiliary elements,
because all subset sums of the auxiliary elements are distinct.
Thus, it suffices to upper-bound the probability that there exist
$I,J\subseteq[n]$ such that
$\sum_{i\in I} w'_i\equiv \sum_{j\in J} w'_j \pmod p$ but
$\sum_{i\in I} w_i\neq \sum_{j\in J} w_j$.
This means that the prime $p$ divides
\[
    D(I,J)=\sum_{i\in I} w_i-\sum_{j\in J} w_j .
\]
We call such primes \emph{bad}.
Since $D(I,J)\neq0$ and $|D(I,J)|\le n2^m$, there are at most
$\log(n2^m)$ bad primes for each fixed pair $I,J\subseteq[n]$.
Since there are $2^{2n}$ possible pairs $I,J\subseteq[n]$, the total
number of bad primes is at most
\[
2^{2n}\log(n2^m)
=
2^{2n}(m+\log n)
\le
2^{4n},
\]
where the last inequality follows from $m<2^n$.

By \cref{lm:prime-density}, the interval $[2^{bn},2^{bn+1}]$ contains
$\Omega(2^{bn}/n)$ primes.
Therefore, the probability that we choose a bad prime from this interval
is upper-bounded by
\[
O\left(\frac{2^{4n}}{2^{bn}/n}\right)
=
O(n2^{-(b-4)n})
\le
O(2^{-n}),
\]
since $b\ge 6$.
Thus, with probability at least $1-O(2^{-n})$, every ESS solution of
$S'$ corresponds to an ESS solution of $S$.
Given an ESS solution of $S'$, we discard the auxiliary elements and
verify the equality in the original instance; if the verification
succeeds, the remaining elements form a corresponding ESS solution of
$S$.
This completes the proof.

\section{Omitted Proofs and Additional Explanations for Section~\ref{sec:FasterExpSpaceAlgorithms}}

Recall that the modified input set $S'$ generated by
\cref{alg:RandomReplace} still satisfies the property guaranteed by
\cref{lm:preprocess-inputs}: all input integers are positive and at most $2^{O(n)}$.

\subsection{Proof of Claim~\ref{clm:Lower-Bound-Phi-mod-p}}
\label{app:Repeat-Proof-Of-ESA-Lemma3.6}
We restate the definition of $\Phi$ and \cref{clm:Lower-Bound-Phi-mod-p}.
\[
\Phi =
\left\{
\sum(X) \;\middle|\;
\begin{aligned}
&X \subseteq S',\ \exists Y \subseteq S' \text{ such that }
X \neq Y,\ \sum(X)=\sum(Y),\\
&X \Delta Y \text{ contains exactly one of } e_{2i-1}
\text{ and } e_{2i}
\text{ for every } i \in [k]
\end{aligned}
\right\}.
\]

\lowerboundphimodp*

We follow the proof of \cite[Lemma~3.6]{mucha_et_al:LIPIcs.ESA.2019.73}.
Let $a_1,a_2\in\Phi$ be two distinct elements.
Then $a_1\equiv a_2\pmod p$ holds if and only if $p$ divides
$|a_1-a_2|$.
By \cref{lm:preprocess-inputs}, we have
$|a_1-a_2|\le 2^{O(n)}$, and hence $|a_1-a_2|$ has only $O(n)$
distinct prime divisors.
On the other hand, by \cref{lm:prime-density}, for sufficiently large
$n$, the interval $[p_{\max},2p_{\max}]$ contains at least
$
\Omega\left(\frac{p_{\max}}{\log p_{\max}}\right)
=
\Omega\left(\frac{p_{\max}}{n}\right)
$
primes, where we use $p_{\max} = 2^{\Theta(n)}$.
Since $p$ is chosen uniformly at random from the primes in
$[p_{\max},2p_{\max}]$, the collision probability for this fixed pair
$a_1,a_2$ is at most $O(n^2/p_{\max})$.

Let
\[
q = \sum_{r = 0}^{p-1} |\{ a \in \Phi \mid a \equiv r \pmod p\}|^2,
\]
which is the number of collisions, including trivial ones.
By linearity of expectation and $p_{\max} \le |\Phi|$, we have
\[
\mathbb{E}[q]
\le
O\left(|\Phi| + |\Phi|^2 \cdot \frac{n^2}{p_{\max}}\right)
\le
O\left(\frac{(|\Phi|n)^2}{p_{\max}}\right).
\]
By Markov's inequality, with constant probability, we have
$
q \le O\left(\frac{(|\Phi|n)^2}{p_{\max}}\right).
$
If this happens, then by the Cauchy--Schwarz inequality,
\[
|\{a \bmod p \mid a \in \Phi\}|
\ge
\frac{|\Phi|^2}{q}
\ge
\Omega(p_{\max}/n^2)
=
\Omega(p/n^2),
\]
where the last follows from $p_{\max} \le p \le 2p_{\max}$.

\subsection{Small-Space Generation of Modular Subset-Sum Lists}
\label{app:Modulo-Schroeppel-Shamir}

\begin{algorithm}[htbp]
\caption{ModuloSmallSpaceGenerator}
\DontPrintSemicolon

\KwIn{Integer lists $A$ and $B$, and modulus $p$}
\KwOut{The multiset
$C = \{\!\{(a+b) \bmod p \mid a \in A,\ b \in B\}\!\}$
generated in ascending order}

Replace each element of $A$ and $B$ by its residue modulo $p$.\;
Sort $B$ in ascending order.\;
Write $A=[a_1,\ldots,a_m]$ and $B=[b_1,\ldots,b_n]$.\;
Initialize an array $D=[d_1,\ldots,d_m]$.\;
Initialize a min-priority queue $Q$.\;

\For{$i=1,\ldots,m$}{
    Let $d_i$ be the smallest index $j$ such that $b_j\ge p-a_i$;
    if no such index exists, set $d_i=1$.\;
    Insert $\bigl((a_i+b_{d_i})\bmod p,\ i,\ 0\bigr)$ into $Q$.\;
}

\While{$Q$ is not empty}{
    Extract $\bigl(s,i,t\bigr)$ from $Q$.\;
    Output $s$.\;
    \If{$t+1<n$}{
        Let $j'=((d_i+t)\bmod n)+1$.\;
        Insert $\bigl((a_i+b_{j'})\bmod p,\ i,\ t+1\bigr)$ into $Q$.\;
    }
}

\label{alg:ModuloSmallSpaceGenerator}
\end{algorithm}

\Cref{lm:SmallSpaceGeneration} is used implicitly in the balanced 4-table algorithm
of Schroeppel and Shamir~\cite{SchroeppelAndShamir1981}.

\begin{lemma}[\cite{SchroeppelAndShamir1981}]
\label{lm:SmallSpaceGeneration}
Given two integer lists $A$ and $B$ of size at most $N$, the multiset
$C = \{\!\{a+b \mid a \in A,\ b \in B\}\!\}$ has size at most $N^2$ and
can be generated in ascending or descending order in $O(N^2\log N)$ time
using only $O(N)$ additional space.
\end{lemma}

In \cref{lm:Enum-C-time-space}, we use a modular version of
\cref{lm:SmallSpaceGeneration}, which generates the values in ascending
or descending order of their residues modulo $p$.
We give pseudocode for the ascending-order version in
\cref{alg:ModuloSmallSpaceGenerator}, obtained by slightly modifying the
implementation of Schroeppel and Shamir~\cite{SchroeppelAndShamir1981}.

The main modification is that we use an additional array $D$ to specify
the cyclic rotation of $B$ for each element of $A$.
Although \cref{alg:ModuloSmallSpaceGenerator} outputs only the generated
values, we can also return the corresponding index pair $(i,j)\in[m]\times[n]$ together
with each value.
The descending-order version can be obtained analogously.
This gives \cref{lm:ModuloSmallSpaceGeneration}.

\begin{lemma}
\label{lm:ModuloSmallSpaceGeneration}
Given two integer lists $A$ and $B$ of size at most $N$, the multiset
$C = \{\!\{(a+b) \bmod p \mid a \in A,\ b \in B\}\!\}$ of size at most $N^2$
can be generated in ascending or descending order in $O(N^2\log N)$ time
using only $O(N)$ additional space.
\end{lemma}

\subsection{Small-Space Algorithm for Unbalanced Equal Subset Sum}
\label{app:SmallSpaceUnbalancedESS}

\begin{algorithm}[t]
\caption{SmallSpaceUnbalancedESS}
\DontPrintSemicolon

\KwIn{Input set $S$, target integer $\kappa$, and minimum solution size $\ell$}
\KwOut{An ESS solution for S if one exists; otherwise, NO}

Randomly split $S$ into four disjoint sets
$S_1, S_2, S_3, S_4 \subseteq S$
such that $|S_i|=n/4$ for all $i \in [4]$.\;
For each $i\in [4]$, enumerate
\begin{minipage}[h]{0.9\linewidth}
\[
D_i =
\left\{\sum(X)-\sum(Y) \mid X,Y\subseteq S_i,\ X\cap Y=\emptyset,\ |X|+|Y|=\ell/4\right\}.
\]
\end{minipage}
\;
\begin{minipage}[t]{0.9\linewidth}
Define $C_1$ and $C_2$ as
\[
\begin{aligned}
C_1
&=
\left\{\sum(X)-\sum(Y) \mid
X,Y\subseteq S_1\cup S_2,\ X\cap Y=\emptyset,\ |X|+|Y|=\ell/2\right\},\\
C_2
&=
\left\{\sum(X)-\sum(Y) \mid
X,Y\subseteq S_3\cup S_4,\ X\cap Y=\emptyset,\ |X|+|Y|=\ell/2\right\}.
\end{aligned}
\]
\end{minipage}\;
Initialize a generator $G_1$ that generates the values of $C_1$ in
ascending order by applying \cref{lm:SmallSpaceGeneration} to $D_1$ and
$D_2$.\;
Initialize a generator $G_2$ that generates the values of $C_2$ in
descending order by applying \cref{lm:SmallSpaceGeneration} to $D_3$ and
$D_4$.\;
Set $x_1$ and $x_2$ to the first values generated by $G_1$ and $G_2$,
respectively.\;
\While{$x_1$ and $x_2$ are both defined}{
    \If{$x_1+x_2=0$}{
        Let $A_1, B_1 \subseteq S_1\cup S_2$ be such that
        $x_1 = \sum(A_1) - \sum(B_1)$.\;
        Let $A_2, B_2 \subseteq S_3\cup S_4$ be such that
        $x_2 = \sum(A_2) - \sum(B_2)$.\;
        \Return{$(A_1\cup A_2,\ B_1\cup B_2)$}\;
    }
    \ElseIf{$x_1+x_2<0$}{
        Advance $G_1$ and update $x_1$ to the next generated value.\;
    }
    \Else{
        Advance $G_2$ and update $x_2$ to the next generated value.\;
    }
}
\Return{\textnormal{NO}}\;
\label{alg:SmallSpaceUnbalancedESS}
\end{algorithm}

We can reduce the space usage of the \emph{UnbalancedEqualSubsetSum}
algorithm of \cite{mucha_et_al:LIPIcs.ESA.2019.73} from
$O^*(\binom{n/2}{\ell/2}2^{\ell/2})$
$=O^*(2^{\frac{n}{2}(H(\ell')+\ell')})$
to
$O^*(\binom{n/4}{\ell/4}2^{\ell/4})$
$=O^*(2^{\frac{n}{4}(H(\ell')+\ell')})$
by applying \cref{lm:SmallSpaceGeneration}.
Recall that $n$ is assumed to be divisible by $12$ in \cref{sec:prelim}.
For simplicity, we further assume here that $\ell$ is divisible by $4$.
The case where $\ell$ is not divisible by $4$ can be handled similarly,
although the floors and ceilings make the notation more cumbersome.
We show the pseudocode in \cref{alg:SmallSpaceUnbalancedESS}.

Let $A,B\subseteq S$ be a fixed minimum solution with $|A\cup B|=\ell$.
The probability that the random four-way split satisfies
$|(A\cup B)\cap S_i|=\ell/4$ for all $i\in[4]$ is
\[
\frac{
\binom{\ell}{\ell/4,\ell/4,\ell/4,\ell/4}
\binom{n-\ell}{(n-\ell)/4,(n-\ell)/4,(n-\ell)/4,(n-\ell)/4}
}{
\binom{n}{n/4,n/4,n/4,n/4}
}
\ge
\frac{4^\ell 4^{n-\ell}}{\binom{n+3}{3}^2 4^n}
=
\frac{1}{\binom{n+3}{3}^2}
\]
where we use
\[
\frac{4^N}{\binom{N+3}{3}}
\le
\binom{N}{N/4,N/4,N/4,N/4}
\le
4^N.
\]
Thus, by repeating the algorithm $\mathrm{poly}(n)$ times, we can amplify
the success probability to a constant without affecting the running time
in $O^*(\cdot)$ notation.
Compared to the \emph{UnbalancedEqualSubsetSum} algorithm of
\cite{mucha_et_al:LIPIcs.ESA.2019.73},
\cref{alg:SmallSpaceUnbalancedESS} stores $D_i$ for $i\in[4]$
instead of $C_i$ for $i\in[2]$.
Thus, the space usage is reduced from $O^*(|C_i|)=O^*(2^{\frac{n}{2}(H(\ell')+\ell')})$
to $O^*(|D_i|)=O^*(2^{\frac{n}{4}(H(\ell')+\ell')})$.
The elements of $C_i$ can be generated in the required ascending or descending
order using \cref{lm:SmallSpaceGeneration}.
This takes $O^*(|D_i|)$ space and $O^*(|C_i|)$ time.

\section{Generalization to Target Equal-Subset-Sum}
\label{app:Target-ESS}
We restate \cref{lm:TargetESS} here.
\targetess*
The meet-in-the-middle algorithm for ESS shown in
\cref{alg:SmallSpaceUnbalancedESS} also works for Target ESS.
The only difference is that we need to determine whether there exist
$x_1\in C_1$ and $x_2\in C_2$ such that $x_1+x_2=\kappa$, rather than
$x_1+x_2=0$.
By replacing the target value $0$ in Lines~8 and~12 of
\cref{alg:SmallSpaceUnbalancedESS} with $\kappa$, we obtain an algorithm solving Target ESS
with the same time and space complexities.

\end{document}